%% file: median-filter.tex
\documentclass[a4paper,11pt]{article}
\pdfoutput=1

\usepackage[vmargin=25mm,hmargin=25mm]{geometry}
\usepackage{amsmath,amssymb,amsthm,booktabs,color,doi,enumitem,framed,graphicx,latexsym,url,verbatim}
\usepackage[numbers,sort&compress]{natbib}
\usepackage[small]{caption}
\usepackage{microtype}
\usepackage{multibib}
\usepackage[low-sup]{subdepth}
\usepackage{hyperref}

\DeclareMathOperator{\prev}{prev}
\DeclareMathOperator{\next}{next}
\DeclareMathOperator{\Construct}{construct}
\DeclareMathOperator{\Delete}{delete}
\DeclareMathOperator{\Undelete}{undelete}
\DeclareMathOperator{\Unwind}{unwind}
\DeclareMathOperator{\Advance}{advance}
\DeclareMathOperator{\Small}{small}
\DeclareMathOperator{\Peek}{peek}

\newcommand{\mydot}[1]{\smash{\dot{#1}}}
\newcommand{\myddot}[1]{\smash{\ddot{#1}}}

\newcites{online}{Online Forums and Code Repositories}
\newcites{software}{Software and Hardware}
\newcommand{\rmfcite}{\citesoftware[p.~1507]{r-manual}}

\newtheorem{lemma}{Lemma}

\hypersetup{
    colorlinks=true,
    linkcolor=black,
    citecolor=black,
    filecolor=black,
    urlcolor=black,
}

\hypersetup{
    pdfauthor={Jukka Suomela},
    pdftitle={Median filtering is equivalent to sorting},
}

\newcommand{\algorithm}[2]{%
\begin{framed}%
\begin{description}[style=nextline,parsep=2pt,itemsep=9pt,#1]%
#2
\end{description}%
\vspace{-4mm}%
\end{framed}%
\vspace{-4mm}%
}

\newcommand{\block}[1]{%
\begin{itemize}[nolistsep,label=,parsep=2pt,topsep=0pt,leftmargin=10mm]%
\item #1%
\end{itemize}%
\vspace{1pt}%
}

\newcommand{\codelabel}[1]{\hfill #1}

\begin{document}

\vspace*{\stretch{1}}
\begin{flushleft}
{\huge\bf Median Filtering is Equivalent to Sorting \par}
\bigskip
\bigskip
\textbf{Jukka Suomela}

Helsinki Institute for Information Technology HIIT,\\
Department of Information and Computer Science, \\ Aalto University, Finland \\
jukka.suomela@aalto.fi
\end{flushleft}

\bigskip
\paragraph{Abstract.}

This work shows that the following problems are equivalent, both in theory and in practice:
\begin{itemize}
    \item \emph{median filtering}: given an $n$-element vector, compute the sliding window median with window size $k$,
    \item \emph{piecewise sorting}: given an $n$-element vector, divide it in $n/k$ blocks of length $k$ and sort each block.
\end{itemize}
By prior work, median filtering is known to be at least as hard as piecewise sorting: with a single median filter operation we can sort $\Theta(n/k)$ blocks of length $\Theta(k)$. The present work shows that median filtering is also as easy as piecewise sorting: we can do median filtering with one piecewise sorting operation and linear-time postprocessing. In particular, median filtering can directly benefit from the vast literature on sorting algorithms---for example, adaptive sorting algorithms imply adaptive median filtering algorithms. 

The reduction is very efficient in practice---for random inputs the performance of the new sorting-based algorithm is on a par with the fastest heap-based algorithms, and for benign data distributions it typically outperforms prior algorithms.

The key technical idea is that we can represent the sliding window with a pair of sorted doubly-linked lists: we delete items from one list and add items to the other list. Deletions are easy; additions can be done efficiently if we \emph{reverse the time} twice: First we construct the full list and delete the items in the reverse order. Then we \emph{undo each deletion} with Knuth's \emph{dancing links} technique.

\thispagestyle{empty}
\vspace*{\stretch{1.3}}
\clearpage
\setcounter{page}{1}

\section{Introduction}\label{sec:intro}

\paragraph{Median filter.}

We study the following problem, commonly known as the \emph{median filter}, \emph{sliding window median}, \emph{moving median}, \emph{running median}, \emph{rolling median}, or \emph{median smoothing}:
\begin{itemize}[noitemsep]
    \item \textbf{Input:} vector $(x_1,x_2,\dotsc,x_n)$ and window size $k$.
    \item \textbf{Output:} vector $(y_1,y_2,\dotsc,x_{n-k+1})$, where $y_i$ is the median of $(x_i,x_{i+1},\dotsc,x_{i+k-1})$.
\end{itemize}
Median filtering and its multidimensional versions are commonly used in digital signal processing \citesoftware{r-manual,mathematica-mf,matlab-mf,octave-mf,scipy-mf} and image processing \citesoftware{photoshop-mf,gimp-mf}; see Figure~\ref{fig:example} for a simple example that demonstrates how efficiently a median filter can recover a corrupted signal.

\paragraph{Contribution.}

This work gives a new, simple and efficient algorithm for median filtering. The new algorithm is based on sorting; there are two phases:
\begin{enumerate}[noitemsep]
    \item \textbf{Piecewise sorting}: divide the input vector in $n/k$ blocks of length $k$, and sort each block.
    \item \textbf{Postprocessing}: compute the output vector in linear time.
\end{enumerate}
If we use a comparison sort, the worst-case running time is $O(n \log k)$, which matches the previous heap-based algorithms~\cite{astola89median,juhola91comparison,hardle95median-smooth}. However, in the new algorithm we can easily plug in any sorting algorithm that exploits the properties of our input vectors (e.g., integer sorting and adaptive sorting), and we can also benefit from sorting algorithms designed for modern computer architectures (e.g., cache-efficient sorting and GPU sorting).

The new algorithm is asymptotically optimal for any reasonable input distribution and model of computing, assuming that we have an optimal sorting algorithm for the same setting. There is a matching lower bound~\cite{juhola91comparison,krizanc05range-mode} that shows that median filtering is at least as hard as piecewise sorting: with a single median filter operation we can sort $\Theta(n/k)$ vectors of length $\Theta(k)$.

The new sorting-based median filter algorithms (with off-the-self sorting algorithm implementations) is \textbf{very efficient in practice} on modern hardware---for random inputs the performance is in the same ballpark as the performance of the best heap-based algorithms, and e.g.\ for partially sorted inputs it typically outperforms the heap-based algorithms by a large factor. Both a simple Python implementation and a highly optimised C++ implementation are available online~\citeonline{suomela-mf}, together with a testing framework and numerous benchmarks that compare the new algorithm with 9 other implementations---including those from \emph{R}, \emph{Mathematica}, \emph{Matlab}, \emph{Octave}, and \emph{SciPy}.

\paragraph{Techniques.}

On a high-level, the postprocessing phase maintains a pair of sorted doubly-linked lists, $L_A$ and $L_B$, so that their union $L_A \cup L_B$ represents the sliding window. Initially, $L_A$ contains the first block of data and $L_B$ is empty. We remove old items from $L_A$ and add new items to $L_B$ until $L_A$ becomes empty and $L_B$ contains the second block of data. We repeat this for each block of input.

To efficiently find the median of $L_A \cup L_B$, we can maintain a pair of pointers, one pointing to $L_A$ and another pointing to $L_B$, and proceed as if we were in the middle of merging two sorted lists.

The key challenge is related to the maintenance of $L_B$. Deletions from a sorted doubly-linked list are easy, but insertions are hard. The key idea is to \textbf{reverse the time}: instead of adding some elements $z_1, z_2, \dotsc, z_k$ to $L_B$ one by one, we start with a list that contains all of these elements and delete them one by one, in the reverse order $z_k, z_{k-1}, \dotsc, z_1$. Now to solve the original problem of adding elements one by one, it is sufficient to \textbf{undo the deletions} one by one. With doubly linked lists, this is very efficiently achieved with Knuth's \textbf{dancing links} technique~\cite{knuth00dancinglinks}.

\begin{figure}
    \centering
    \includegraphics{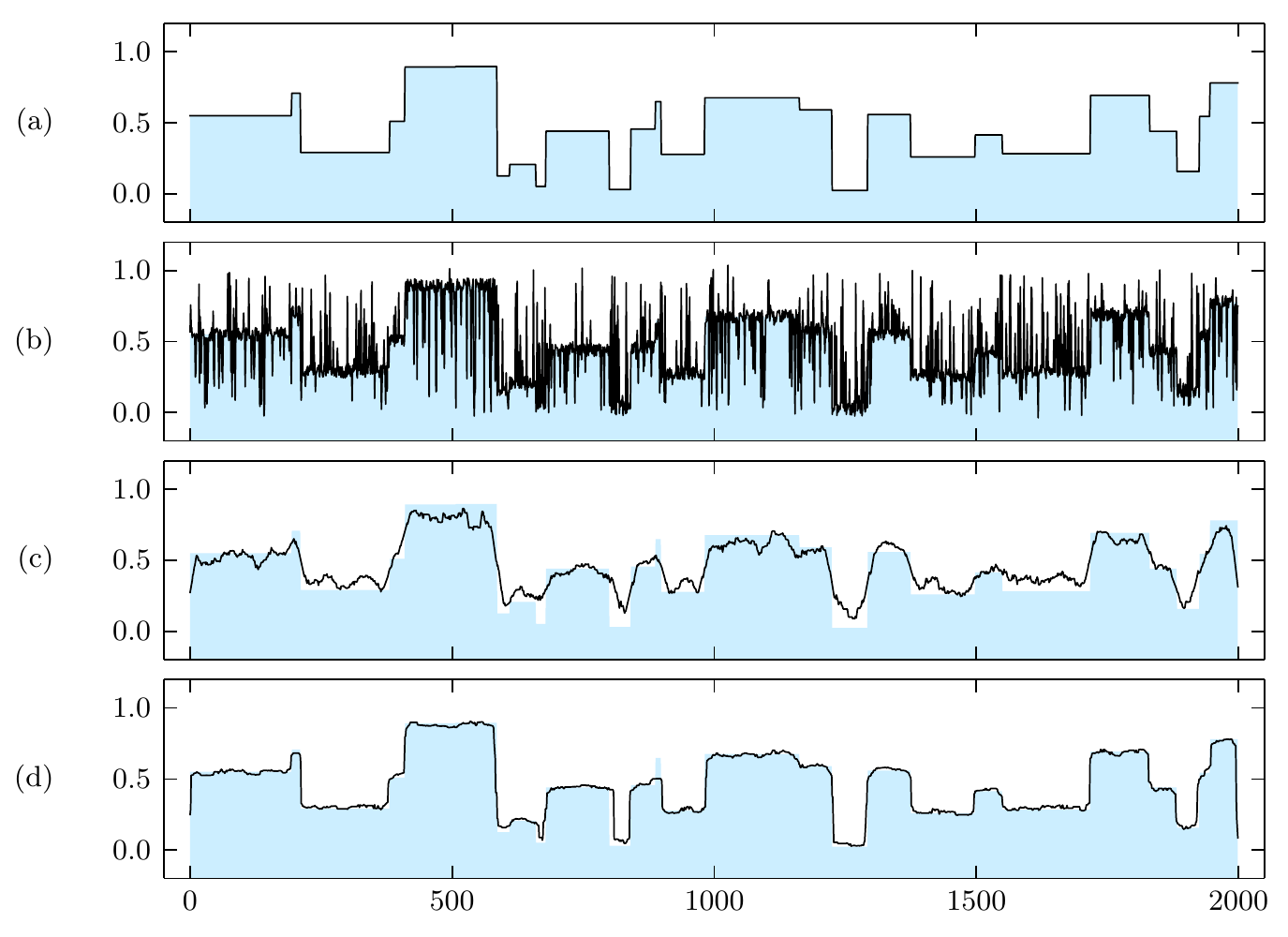}
    \caption{The median filter can recover corrupted data much better than e.g.\ moving average filters. (a)~Original data, $n=2000$. (b)~25\% of data points corrupted, some random noise added. (c)~Moving average filter applied, window size $k=25$. (d)~Median filter applied, window size $k=25$. In all figures, the shaded area represents the original data.}\label{fig:example}
\end{figure}

\section{Prior Work}

\paragraph{Algorithms.}

There is, of course, a trivial algorithm for median filtering in time $O(n k)$: simply \emph{find the median separately for each window}. This approach, together with sorting networks, can be attractive for hardware implementations of median filters~\cite{oflazer83median}, but as a general-purpose algorithm it is inefficient.

Non-trivial algorithms presented in the literature are unanimously based on the following idea: \emph{maintain a data structure that represents the sliding window}. Such a data structure needs to support three operations: ``construct'', ``find the median'', and ``remove the oldest element and add a new element''. With such a data structure, one can first construct it with elements $x_1, x_2, \dotsc, x_k$, and then process elements $x_{k+1}, x_{k+2}, \dotsc, x_n$ one by one, in this order. Concrete ideas for the implementation of the window data structure can be classified as follows:
\begin{enumerate}
    \item Data structures for $B$-bit integers. For a small $B$, we can easily maintain a histogram with $2^B$ buckets. However, to find the new median we need to find an adjacent unoccupied bucket. The following approaches have been discussed in the literature:
    \begin{enumerate}[noitemsep]
        \item linear scanning \cite{huang79median,ateman80median,juhola91comparison}: worst-case running time $\Theta(n 2^B)$
        \item binary trees \cite{ateman80median,juhola91comparison}: worst-case running time $\Theta(n B)$
        \item van Emde Boas trees \cite{juhola91comparison}: worst-case running time $\Theta(n \log B)$.
    \end{enumerate}
    \item Efficient comparison-based data structures with a $\Theta(n \log k)$ worst-case running time:
    \begin{enumerate}[noitemsep]
        \item a maxheap-minheap pair \cite{astola89median,juhola91comparison,hardle95median-smooth}
        \item binary search trees \cite{juhola91comparison}
        \item finger trees \cite{juhola91comparison}.
    \end{enumerate}
    \item Inefficient comparison-based data structures with a $\Theta(n k)$ worst-case running time:
    \begin{enumerate}[noitemsep]
        \item doubly-linked lists \cite{juhola91comparison}
        \item sorted arrays \cite{ahmad87median,juhola91comparison}.
    \end{enumerate}
\end{enumerate}
In summary, the search for efficient median filter algorithms has focused on the design of an efficient data structure for the sliding window. While it is known that 2-dimensional median filtering can benefit from a clever traversal order \cite{perreault07median}, it seems that all existing algorithms for 1-dimensional median filtering are based on the idea of a doing a single, uniform, in-order traversal of the input vector.

It seems that the present work is the first deviation from this trend in the long history of median filtering algorithms. In essence, we see median filtering as an algorithmic challenge---instead of asking how to construct an efficient \emph{data structure} for the sliding window, we ask how to \emph{pre-process} the input vector so that the sliding window is much easier to maintain.

\paragraph{Applications and Implementations.}

Median filtering has been applied in statistical data analysis at least since 1920s \cite{king24seasonal}, and it was popularised by Tukey in 1970s \cite[Section 7A]{tukey77eda}.

Nowadays, a median filter is a standard subroutine in numerous scientific computing environments and signal processing packages. In \emph{R} it is called ``runmed'' \rmfcite{}, and in \emph{Mathematica} it is called ``MedianFilter'' \citesoftware{mathematica-mf}. \emph{Matlab}'s Signal Processing Toolbox, \emph{GNU Octave}'s ``signal'' package, and \emph{SciPy}'s module ``scipy.signal'' all provide a median filter function called ``medfilt1'' \citesoftware{scipy-mf,matlab-mf,octave-mf}.

Multidimensional generalisations of the median filter are commonly used in image processing. For example, in \emph{Photoshop} there is a noise reduction filter called ``Median'' \citesoftware{photoshop-mf}, and in \emph{Gimp} there is a ``Despeckle'' filter, which is a generalisation of the 2-dimensional median filter \citesoftware{gimp-mf}.

Surprisingly, most of the existing implementations of the median filter in scientific computing environments are very inefficient for a large $k$. The experiments conducted in this work demonstrate that the median filter functions in the current versions of \emph{Matlab}, \emph{Mathematica}, \emph{Octave}, and \emph{SciPy} all exhibit approximately $\Theta(n k)$ complexity for random inputs (see Figure~\ref{fig:othera} for examples). It should be noted that these software packages typically provide very efficient routines for sorting, which would make the algorithm presented in this work relatively easy to implement.

The only major software package with an efficient $\Theta(n \log k)$ median filter implementation seems to be \emph{R}\@. For large values of $k$, the runmed function in \emph{R} applies a high-quality implementation of the double-heap data structure \cite{astola89median,juhola91comparison,hardle95median-smooth}. The end result is very efficient both in theory and in practice, for a wide range of $n$ and~$k$ (see Figures \ref{fig:othera} and \ref{fig:otherb} for examples).

We are aware of only one general-purpose median filter implementation that consistently outperforms \emph{R}: an open source C implementation by AShelly from 2011~\citeonline{ashelly-answer,ashelly-code}. This is, again, an implementation of the double-heap technique. Raffel~\citeonline{raffel-code} has adapted this implementation to C++, and we will use Raffel's version as a baseline in our experiments.

\paragraph{Lower Bounds.}

There is a simple argument that shows that median filtering is at least as difficult as piecewise sorting---see, e.g., Juhola et al.~\cite{juhola91comparison} and Krizanc et al.~\cite{krizanc05range-mode}. Assume that $k = 2h+1$, and assume that we want to sort $n / (3h+2)$ blocks of size $h+1$. Construct the input vector $x$ so that before each block we have $h$ times the value $-\infty$ and after each block we have $h$ times the value $+\infty$. If we now apply the median filter, it is easy to see that in the output each block is sorted.

Hence with some linear-time preprocessing and postprocessing, and $O(1)$ invocations of the median filter operation, we can sort $n/k$ blocks of length $k$. This work shows that the converse is also true.

\section{Algorithm Overview}

Figure~\ref{fig:overview} provides an illustration of the key definitions and the overall behaviour of the algorithm.

\begin{figure}
    \centering
    \includegraphics[page=1]{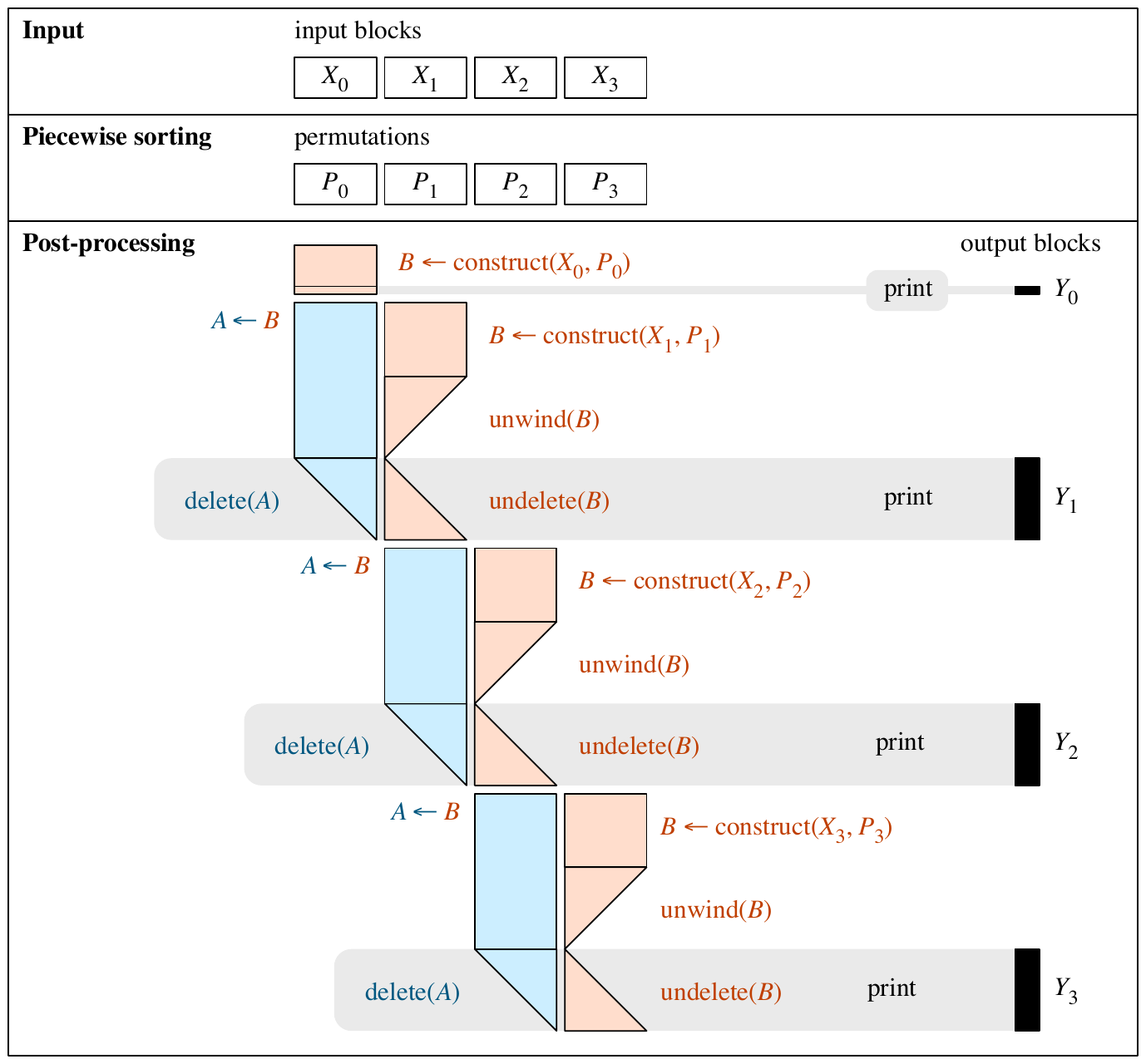}
    \caption{Algorithm overview.}\label{fig:overview}
\end{figure}

\paragraph{Preliminaries.}
To keep the presentation easy to follow, we will assume that $n = bk$ and $k = 2h + 1$, for integers $h$ and $b$. Here $h$ is the size of a \emph{half-window} and $b$ is the number of \emph{blocks}. Extending the algorithm to arbitrary $n$ and $k$ is straightforward.

Throughout this work, we use arrays with bracket notation and \textbf{0-based indexing}: if $\alpha$ is an array of length $k$, then its elements are $\alpha[0], \alpha[1], \dotsc, \alpha[k-1]$. We will partition input vector $x$ and output vector $y$ in $b$ arrays of length $k$, as follows:
\newcommand{\myunderbrace}[2]{\underbrace{#1}_{\displaystyle #2}}
\begin{gather*}
    \myunderbrace{x_1, x_2, \dotsc, x_k}{X_0},
    \ \myunderbrace{x_{k+1}, x_{k+2}, \dotsc, x_{2k}}{X_1},
    \ \dotsc,
    \ \myunderbrace{x_{bk-k+1}, x_{bk-k+2}, \dotsc, x_{bk}}{X_{b-1}}\,, 
    \\[5pt]
    \myunderbrace{\bot, \dotsc, \bot, y_1}{Y_0},
    \ \myunderbrace{y_{2}, y_{3}, \dotsc, y_{k+1}}{Y_1},
    \ \dotsc,
    \ \myunderbrace{y_{bk-2k+2}, y_{bk-2k+3}, \dotsc, y_{bk-k+1}}{Y_{b-1}}.
\end{gather*}
Here we use the symbol $\bot$ for padding.

\pagebreak

\paragraph{Piecewise Sorting.}

For each $j$, find a permutation $P_j$ of $\{0,1,\dotsc,k-1\}$ that sorts the elements of array $X_j$. That is, for all $0 \le s < t < k$ we have $X_j[P_j[s]] \le X_j[P_j[t]]$.

\paragraph{Postprocessing.}

The first output array $Y_0$ is trivial to compute: its only element is $Y_0[k-1] = X_0[P_0[h]]$. Let us now focus on the case of $1 \le j < b$. Define
\[
    \alpha_A = X_{j-1}, \quad
    \alpha_B = X_{j}, \quad
    \pi_A = P_{j-1}, \quad
    \pi_B = P_{j}, \quad
    \beta = Y_j.
\]
We will show how to find $\beta$ in time $O(k)$ given $\alpha_A$, $\alpha_B$, $\pi_A$, and $\pi_B$.

The basic idea is simple: We maintain \emph{sorted doubly-linked lists} $L_A$ and $L_B$, so that their union $L_A \cup L_B$ represents the sliding window. Initially, $L_A$ contains the elements of block $\alpha_A$ in an increasing order while $L_B$ is empty. At each time step $t = 0,1,\dotsc,k-1$, we remove element $\alpha_A[t]$ from $L_A$, and add element $\alpha_B[t]$ to $L_B$---we will shortly see how to do this efficiently. In the end, $L_A$ will be empty and $L_B$ will contain the elements of $\alpha_B$ in an increasing order. We augment the data structures $L_A$ and $L_B$ with additional pointers so that we can efficiently find the median of $L_A \cup L_B$ after each time step $t$.

The key challenge is related to the maintenance of the linked lists $L_A$ and $L_B$. At first, there seems to be inherent asymmetry:
\begin{itemize}[noitemsep]
    \item Maintenance of $L_A$ is easy: we only need to remove elements from a doubly-linked list.
    \item Maintenance of $L_B$ is hard: we have to add elements in the right place to keep $L_B$ sorted.
\end{itemize}
The key insight is that the situation is symmetric with respect to time.

\section{Main Ingredient: Time Reversal and Dancing Links}

Recall that our goal is to efficiently solve the following task, so that at each point list $L$ is a sorted doubly-linked list:
\begin{enumerate}[label=(P\arabic*),leftmargin=15mm]
    \item\label{P1} insert $\alpha[0], \alpha[1], \dots, \alpha[k-1]$ into $L$ one by one
\end{enumerate}
If we \emph{reverse the time}, our original process becomes
\begin{enumerate}[resume*]
    \item\label{P2} delete $\alpha[k-1], \alpha[k-2], \dots, \alpha[0]$ from $L$ one by one.
\end{enumerate}
Finally, to recover the original process, we reverse the time again, obtaining
\begin{enumerate}[resume*]
    \item\label{P3} undo the deletions of $\alpha[0], \alpha[1], \dots, \alpha[k-1]$ one by one.
\end{enumerate}
While \ref{P1} looks difficult to implement, \ref{P2} is easy to solve, and then \ref{P3} can be solved with Knuth's dancing links technique~\cite{knuth00dancinglinks}.

We will now explain this idea in more detail. Let us first fix the representation that we will use for linked lists. For list $L$, we will maintain two arrays of indexes, `$\prev$' and `$\next$'. If $\alpha[i]$ is in list $L$, then $\prev[i]$ is the index of the predecessor of $\alpha[i]$ and $\next[i]$ is the successor of $\alpha[i]$.

Given a permutation $\pi$ that sorts $\alpha$, we can easily initialise $\prev$ and $\next$ so that $L$ contains all elements of $\alpha$ in a sorted order; this takes $O(k)$ time. Deletions are also easy: to delete element $\alpha[i]$ from $L$, we simply set
\begin{align}
    \label{eq:del}
    \prev[\next[i]] \gets \prev[i], \quad \next[\prev[i]] \gets \next[i].
\end{align}
Knuth's~\cite{knuth00dancinglinks} observation is that \eqref{eq:del} is easy to reverse:
\begin{align}
    \label{eq:undo-del}
    \prev[\next[i]] \gets i, \quad \next[\prev[i]] \gets i.
\end{align}
In essence, index $i$ and pointers $\prev[i]$ and $\next[i]$ contain enough information to perfectly undo the deletion of $\alpha[i]$ from list $L$.

\pagebreak

Hence we can do the following:
\begin{enumerate}[noitemsep]
    \item \textbf{Construct} the sorted list $L$ with the help of permutation $\pi$.
    \item \textbf{Unwind} the list by deleting $\alpha[k-1], \alpha[k-2], \dots, \alpha[0]$ in this order. Now list $L$ is empty.
    \item At each time step $t = 0, 1, \dotsc, k-1$, \textbf{undo the deletion} of element $\alpha[t]$. In effect, we insert $\alpha[t]$ in the sorted doubly-linked list $L$ in the right position.
\end{enumerate}

The simple idea of combining piecewise sorting, time reversals, and dancing links is all that it takes to design an efficient median filter algorithm. The rest of this work presents the algorithm in more detail.

\begin{figure}
    \centering
    \includegraphics[page=2]{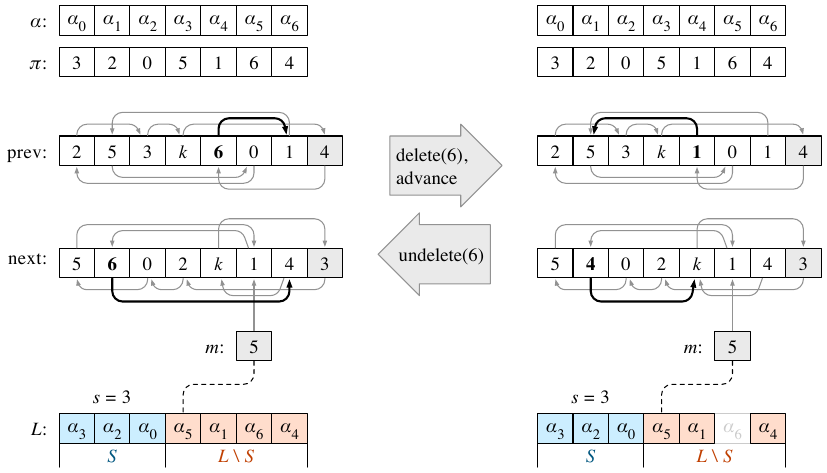}
    \caption{An example of the behaviour of the block data structure.}\label{fig:block}
\end{figure}

\section{Block Data Structure}\label{sec:block}

The algorithm relies on \emph{block data structures} (see Figure~\ref{fig:block}). Conceptually, a block data structure $B$ is a tuple $(\alpha_B,\pi_B,L_B,s_B)$, where array $\alpha_B$ is one block of input, array $\pi_B$ is the permutation that sorts $\alpha_B$, list $L_B$ contains some subset of the elements of $\alpha_B$, and $s_B$ is a counter between $0$ and $|L_B|$. We say that the first $s_B$ elements of list $L_B$ are \emph{small}, and the rest of the elements are \emph{large}. We will omit subscript $B$ when it is clear from the context.

When a block data structure is created, list $L_B$ will contain all $k=2h+1$ elements of $\alpha_B$, and the first $h$ of them will be small. We can then delete elements, undo deletions, and adjust~$s_B$.

\begin{figure}
\input{block-if}
\caption{Block data structure interface.}\label{fig:block-if}
\end{figure}

\subsection{Interface}

The block data structure $B$ supports the operations shown in Figure~\ref{fig:block-if}. The time complexity of $\Construct$ and $\Unwind$ is $O(k)$, and for all other operations it is $O(1)$. Deletions and undeletions must be \emph{properly nested}. For example, this sequence of operations is permitted:
\[
    \Delete(B,15), \Delete(B,3), \Undelete(B,3), \Undelete(B,15).
\]
However, this sequence of operations is not permitted:
\[
    \Delete(B,15), \Delete(B,3), \Undelete(B,15), \Undelete(B,3).
\]

\subsection{Assumption: Stable Sorting}\label{ssec:stable}

For convenience, we will assume that permutation $\pi$ is a stable sort of input $\alpha$. In practice, we can very efficiently find such a $\pi$ as follows: construct an array of pairs $(\alpha[i],i)$, sort it in lexicographic order with any sorting algorithm, and then pick the second element of each pair. This way we have constructed $\pi$ and also guaranteed stability.

We could also do without a stable sort if we slightly modified the algorithm. In essence, we could first construct the inverse permutation $\pi^{-1}$ and then use $\pi^{-1}[i]$ instead of $(\alpha[i],i)$ in comparisons.

\begin{figure}
\input{block-impl}
\caption{Block data structure implementation.}\label{fig:block-impl}
\end{figure}

\subsection{Implementation}

To implement the block data structure $B$, we will use the following fields in addition to input $\alpha$, permutation $\pi$, and counter $s$ (see Figure~\ref{fig:block}):
\begin{itemize}[noitemsep]
    \item $\prev$, $\next$: arrays of length $k+1$,
    \item $m$: integer between $0$ and $k$.
\end{itemize}
Assume that $L = (\alpha[p_0], \alpha[p_1], \dotsc, \alpha[p_{c-1}])$. For convenience, let $p_{-1} = p_c = k$. We will maintain the following invariants:
\begin{itemize}[noitemsep]
    \item $\next[k] = p_0$ and $\next[p_i] = p_{i+1}$ for all $i$,
    \item $\prev[k] = p_{c-1}$ and $\prev[p_i] = p_{i-1}$ for all $i$,
    \item $m = p_s$.
\end{itemize}
Define $\alpha[k] = +\infty$. Given any index $i$ with $\alpha[i] \in L$, it is easy to check if $\alpha[i]$ is small: see if
$
    (\alpha[i], i) < (\alpha[m], m)
$
in lexicographic order---recall that we assumed that this is compatible with permutation $\pi$.

We are now ready to explain how to implement each operation; the algorithm is given in Figure~\ref{fig:block-impl}. While some care is needed in the corner cases (e.g., $m = k$ or $i = m$), it is relatively easy to verify that the invariants are maintained and that the implementation is correct.

\section{Complete Algorithm}\label{sec:alg}

We will now present the complete sorting-based median filter algorithm. Recall that $n = bk$ and $k = 2h + 1$. The input vector $x$ is partitioned in arrays $X_0, X_1, \dotsc, X_{b-1}$. 

\subsection{Preprocessing}

For each $j$, find a permutation $P_j$ that sorts the elements of $X_j$. As discussed in Section~\ref{ssec:stable}, we will assume a stable sort.

\subsection{Postprocessing}

\newcommand{\stepa}{$(\ddagger)$}
\newcommand{\stepc}{$(\rotatebox[origin=c]{180}{$\dagger$})$}
\newcommand{\stepb}{$(\dagger)$}

The algorithm for the postprocessing phase is given in Figure~\ref{fig:alg}. It prints the elements of the output vector $y$ one by one. Recall that Figure~\ref{fig:overview} gives an illustration of the behaviour of the algorithm, and the block data structures $A$ and $B$ were defined in Section~\ref{sec:block}.

\begin{figure}
\input{postprocess}
\caption{Algorithm for median filtering: postprocessing phase.}\label{fig:alg}
\end{figure}

\subsection{Correctness}

Assuming that each block is sorted in a stable manner, the block data structures preserve the order of equal elements. The algorithm of Figure~\ref{fig:alg} then preserves the order of equal elements between blocks: in the case of ties, the elements of $A$ are considered to be smaller than the elements of $B$. Hence for the purposes of the analysis, we can w.l.o.g.\ \textbf{assume that all elements are distinct}---the algorithm behaves precisely as if we had originally broken ties with element indexes. In particular, for a block data structure $B$, we can conveniently interpret $L_B$ as a set. We will write $S_B \subseteq L_B$ for the set of small elements; hence $|S_B| = s_B$.

Let us now turn our attention to the algorithm of Figure~\ref{fig:alg}. Step \stepa{} clearly outputs the median of the first block. In the inner loop, $L_A \cup L_B$ correctly represents the sliding window (cf.\ Figure~\ref{fig:overview}). To show that the algorithm is correct, it is therefore sufficient to show that step \stepc{} outputs the median of the sliding window $L_A \cup L_B$.

Let $H_{AB}$ denote the $h$ smallest items of $L_A \cup L_B$. We will maintain the following invariant: before steps \stepb{} and \stepc, we have
\begin{equation}\label{eq:inv}
    s_A + s_B = h, \quad
    S_A \cup S_B = H_{AB}.
\end{equation}
If \eqref{eq:inv} holds, then the median of $L_A \cup L_B$ is the smallest large item of $L_A$ or $L_B$, and this is precisely what we print in step~\stepc.

We will now argue that the invariant indeed holds throughout the algorithm. Let us first make some easy observations:
\begin{enumerate}
    \item Invariant \eqref{eq:inv} holds before step \stepb{} for iteration $j=1$ and $i=0$.
    \item Assume that invariant \eqref{eq:inv} holds after step \stepc{} for some iteration $j=j_0$ and $i=i_0 < k-1$. Then it holds before step \stepb{} for iteration $j=j_0$ and $i=i_0+1$.
    \item Assume that invariant \eqref{eq:inv} holds after step \stepc{} for some iteration $j=j_0 < b-1$ and $i=k-1$. Then it holds before step \stepb{} for iteration $j=j_0+1$ and $i=0$.
\end{enumerate}
The nontrivial part is covered in the following lemma.

\begin{figure}[p]
    \centering
    \includegraphics[page=3]{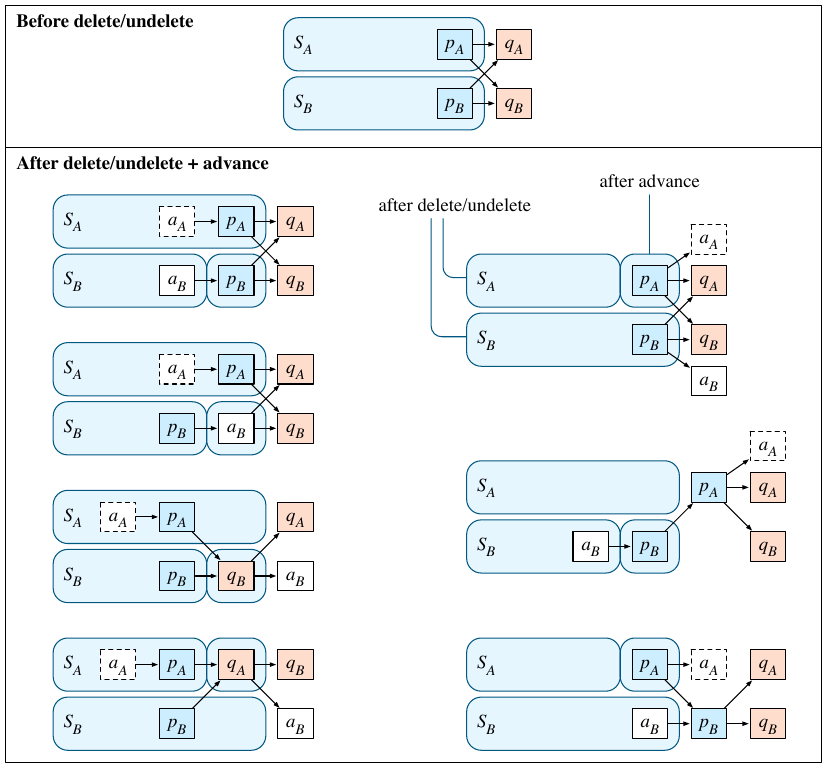}
    \caption{Case analysis in the proof of Lemma~\ref{lem:inv}. Here $a_A$ is the element that we delete from $A$ with the $\Delete(A,i)$ operation, and $a_B$ is the element that we add to $B$ with the $\Undelete(B,i)$ operation. Finally, we apply either $\Advance(A)$ or $\Advance(B)$. See also Table~\ref{tab:inv}.}\label{fig:inv}
\end{figure}

\begin{table}[p]
    \centering
    \input{invariant}
    \caption{Case analysis in the proof of Lemma~\ref{lem:inv}. We use the shorthand notation $U + e = U \cup \{e\}$ and $U - e = U \setminus \{e\}$. See Figure~\ref{fig:inv} for illustrations.}\label{tab:inv}
\end{table}

\begin{lemma}\label{lem:inv}
    Assume that invariant \eqref{eq:inv} holds before step \stepb{} for some iteration $j=j_0$ and $i=i_0$. Then it holds before step \stepc{} for the same iteration $j=j_0$ and $i=i_0$.
\end{lemma}
\begin{proof}
    We will use the following convention to refer to the states of block data structures $A$ and~$B$:
    \begin{itemize}[noitemsep]
        \item $A$ and $B$ to refer to the original states before step \stepb{}, 
        \item $\mydot A$ and $\mydot B$ to refer to the new states after $\Delete$ and $\Undelete$ operations,
        \item $\myddot A$ and $\myddot B$ to refer to the new states before step \stepc{}.
    \end{itemize}

    First, assume that $S_A = \emptyset$. Then all elements of $L_A$ are strictly larger than any element of $S_B = H_{AB}$. If $\alpha_B[i]$ is large, $\Undelete(B,i)$ does not change $S_B$; if $\alpha_B[i]$ is small, $\Undelete(B,i)$ replaces the largest element of $S_B$ with $\alpha_B[i]$. In both cases, $s_{\mydot B} = h$ and therefore $\myddot A = \mydot A$ and $\myddot B = \mydot B$. We conclude that
    \[
        S_{\myddot A} = \emptyset, \quad S_{\myddot B} = H_{\myddot A \myddot B}.
    \]

    Second, assume that $S_A \ne \emptyset$. In this case the $\Delete(A,i)$ operation decreases the size of $S_A$, and we have
    \[
        s_{\mydot A} = s_A - 1, \quad
        s_{\mydot B} = s_B, \quad
        s_{\mydot A} + s_{\mydot B} = h - 1.
    \]
    Hence we will perform one $\Advance$ operation, after which $s_{\myddot A} + s_{\myddot B} = h$. However, it is not entirely obvious that this results in $S_{\myddot A} \cup S_{\myddot B} = H_{\myddot A \myddot B}$, too. To prove this, some case analysis is needed. The critical elements are
    \begin{align*}
        a_A &= \alpha_A[i], &
        p_A &= \max S_A, &
        q_A &= \min (L_A \setminus S_A), \\
        a_B &= \alpha_B[i], &
        p_B &= \max S_B, &
        q_B &= \min (L_B \setminus S_B).
    \end{align*}
    By definition, $\max\{p_A,p_B\} < \min\{q_A,q_B\}$, but this still leaves us with a large number of possible total orders of $\{a_A,p_A,q_A,a_B,p_B,q_B\}$ that we need to consider---also note that $a_A$ may be equal to $p_A$ or $q_A$. However, there are only seven cases that are essentially different from the perspective of what the algorithm does. The seven cases are illustrated in Figure~\ref{fig:inv} and listed in Table~\ref{tab:inv}---each of them corresponds to a partial order of $\{a_A,p_A,q_A,a_B,p_B,q_B\}$, and these partial orders together cover all possible total orders. We can verify that in each case
    \[
        S_{\mydot A} \subseteq H_{\mydot A B}, \quad
        S_{\mydot A} \subseteq H_{\mydot A \mydot B}, \quad
        S_{\mydot B} \subseteq H_{\mydot A \mydot B}, \quad
        S_{\myddot A} \cup S_{\myddot B} = H_{\myddot A \myddot B}.
        \qedhere
    \]
\end{proof}

\subsection{Implementations}

Two implementations of the sorting-based median algorithm are available online~\citeonline{suomela-mf}:
\begin{enumerate}[noitemsep]
    \item A simple Python implementation that is easy to follow.
    \item A highly optimised C++11 version.
\end{enumerate}
A compact version of the Python implementation without comments and assertions is also reproduced in Appendix~\ref{app:python}.

\section{Experiments}

We will now present the experiments in which we compare the performance of the new sorting-based median filter algorithm with \textbf{9 other implementations} of median filter algorithms, including the median filter functions from \emph{R}, \emph{Matlab}, \emph{GNU Octave}, \emph{SciPy}, and \emph{Mathematica}. It turns out that our sorting-based median filter algorithm performs consistently very well in comparison with the other implementations.

For a broad range of window sizes (between $h = 10$ and $h = 5\cdot10^7$) and for various input distributions, our implementation never loses by more than 20~\% in comparison with the fastest median filter algorithm from prior work. In many cases, our algorithm outperforms all competing implementations by a large factor---by a factor up to $3$ for uniform random inputs and by a factor up to $8$ for more benign input distributions.

\subsection{Implementations}\label{ssec:impl}

We will now describe the 11 implementations that we benchmarked. We start with our new algorithm and two simple baseline algorithms---all of these are optimised C++11 implementations:
\begin{description}
    \item[SortMedian:] The sorting-based median algorithm described in this work. For sorting, we use std::sort from the C++ standard library.
    \item[TreeMedian:] The sliding window is maintained as a pair of balanced search trees. We use std::multiset from the C++ standard library---this is typically a highly optimised implementation of a red-black tree.
    \item[MoveMedian:] The sliding window is maintained as a sorted array. Binary search is used to locate the part of the array that needs to be moved in order to accommodate the new element. Standard library routines std::copy and std::copy\_backward are used to efficiently move a block of data.
\end{description}
We have also included an efficient open source median filter implementation in our testing framework---while the algorithm idea dates back to 1980s, this is a modern C++ implementation from 2011:
\begin{description}
    \item[HeapMedian:] The sliding window is maintained as a double heap~\cite{astola89median,juhola91comparison,hardle95median-smooth}. This is Raffel's adaptation~\citeonline{raffel-code} of AShelly's implementation~\citeonline{ashelly-answer,ashelly-code}, with very minor modifications.
\end{description}
The source code of the above algorithms, as well as a unified testing framework, is available online~\citeonline{suomela-mf}. To ensure correctness, there is also a verification tool that compares the outputs of all four implementations against each other.

In addition to these C++ implementations, we also benchmark median filter routines that are available in the following scientific computing environments and signal processing packages:
\begin{itemize}[noitemsep]
    \item \emph{R} \citesoftware{r}, a free software for statistical computing,
    \item \emph{Matlab} \citesoftware{matlab}, a commercial numerical computing environment,
    \item \emph{GNU Octave} \citesoftware{octave}, a free numerical computing environment,
    \item \emph{SciPy} \citesoftware{scipy}, a collection of Python modules for scientific computing,
    \item \emph{Mathematica} \citesoftware{mathematica}, a commercial symbolic computing environment.
\end{itemize}
In total, six algorithm implementations were benchmarked:
\begin{description}
    \item[\emph{R}, runmed(``Turlach''):] The standard routine ``runmed'' \rmfcite{} in \emph{R}, with parameter ``algorithm'' set to ``Turlach''. This implementation maintains the sliding window as a double heap.
    \item[\emph{R}, runmed(``Stuetzle''):] As above, but with parameter ``algorithm'' set to ``Stuetzle''. This implementation maintains the sliding window as a sorted array.
    \item[\emph{Octave}, medfilt1:] Function ``medfilt1'' \citesoftware{octave-mf} in \emph{GNU Octave}'s ``signal'' package. Based on the source code, this function maintains a sorted array.
    \item[\emph{Matlab}, medfilt1:] Function ``medfilt1'' \citesoftware{matlab-mf} in \emph{Matlab}. Based on the source code, this function finds the median separately for each possible location of the sliding window.
    \item[\emph{SciPy}, scipy.signal.medfilt:] Function ``scipy.signal.medfilt'' \citesoftware{scipy-mf} in \emph{SciPy}. Based on the source code, this function finds the median by sorting the sliding window.
    \item[\emph{Mathematica}, MedianFilter:] Function ``MedianFilter'' \citesoftware{mathematica-mf} in \emph{Mathematica}. No source code or information on the algorithm is publicly available.
\end{description}
Finally, to be fair with software packages that rely on median filter implementations written in high-level languages, we also tested a \emph{very slow} implementation of our sorting-based algorithm:
\begin{description}
    \item[SortMedian.py:] The simple Python implementation from Appendix~\ref{app:python}.
\end{description}

\subsection{Comparison of All Implementations}\label{ssec:comp-other}

We will start with a broad comparison of all 11 implementations described in Section~\ref{ssec:impl}. The experiments were conducted as follows (with a few exceptions, see below):
\begin{itemize}
    \item We keep $bh$ fixed and vary $h$. This is approximately equivalent to keeping the size of input vector $n = (2h+1)b$ fixed and varying the window size $k = 2h+1$.
    \item Input consists of independent, uniformly distributed, random 32-bit integers, or its closest equivalent that is supported in the computing environment that we benchmark.
    \item Each experiment was ran 10 times with different random seeds.
    \item The plots report the median running times.
    \item The experiments were ran on the same OS X computer equipped with a 1.7 GHz Intel Core i7 processor and 8 GB of RAM.
\end{itemize}
The following exceptions were made:
\begin{itemize}
    \item \emph{Mathematica}: Only 1 experiment was ran, as this was by far the slowest implementation.
    \item \emph{Matlab}: This implementation required huge amounts of memory. In the end, we resorted to a high-end Linux computer equipped with a 2.8 GHz Intel Xeon processor and 256 GB of RAM. Only 1 experiment was ran, as this is clearly not among the fastest implementations.
    \item \emph{R}: The running times of the fastest experiments (below 10\,ms) are averages of 10 or 100 trials.
\end{itemize}
Detailed information on the software versions and computing platforms is given in Table~\ref{tab:versions}. The source code of the test suite and the raw test results are available online~\citeonline{suomela-mf}.

\begin{table}
    \centering
    \input{versions}
    \caption{The software versions and platforms used in the experiments of Section~\ref{ssec:comp-other} and Figures \ref{fig:othera}--\ref{fig:otherb}.}\label{tab:versions}
\end{table}

First we ran experiments with small parameter values $bh = 10^4$ and $bh = 10^5$ for all implementations; the results are reported in Figure~\ref{fig:othera} in the appendix. From the log-log plots it is easy to see that most of the implementations exhibit running times that are approximately proportional to $nk$. Only four implementations provide a decent performance and scalability: SortMedian, HeapMedian, TreeMedian, and \emph{R}'s ``Turlach'' implementation.

Then we repeated the experiments with the most promising implementations for larger parameter values $bh = 10^6$ and $bh = 10^7$. The results are reported in Figure~\ref{fig:otherb}. The key finding is that SortMedian and HeapMedian consistently outperform all other implementations for large inputs. In the next section, we will focus on these two implementations.

\subsection{Comparison of HeapMedian and SortMedian}

We will now do a more detailed comparison of the fastest algorithms, HeapMedian and SortMedian. These tests were conducted as follows:
\begin{itemize}
    \item We keep $bh$ fixed and vary $h$. In total, we use 66 different combinations of $h$ and~$b$.
    \item We use 2 different versions of the implementations: one for 32-bit inputs and one for 64-bit inputs.
    \item We use 7 different generators to produce the input array $x$:
    \begin{enumerate}[noitemsep]
        \item \emph{asc}: ascending values, $x[i] = i$.
        \item \emph{desc}: descending values, $x[i] = n - i$.
        \item \emph{r-asc}: ascending values + small uniform random noise, $i \le x[i] < i+10^4$.
        \item \emph{r-desc}: descending values + small uniform random noise, $n-i \le x[i] < n-i+10^4$.
        \item \emph{r-large}: large uniform random integers (32-bit or 64-bit).
        \item \emph{r-small}: small uniform random integers, $0 \le x[i] < 10^4$.
        \item \emph{r-block}: piecewise constant data + small uniform random noise.
    \end{enumerate}
    \item For each combination of a version and a generator, we run the experiment for 5 times, with different random seeds.
    \item The plots report the median (solid curve) and the region from the 2nd decile to the 9th decile (shading). That is, the shaded area represents 80~\% of the experiments.
    \item The experiments were ran on Linux, using the Intel Xeon nodes of the Triton cluster \citesoftware{triton}, with one processor allocated for each experiment.
    \item To compile the code, we used GCC version 4.8.2 and GCC's implementation of the C++ standard library (a.k.a.\ libstdc++).
\end{itemize}
In total, this setup results in $66 \times 2 \times 7 \times 5 = 4620$ experiments per algorithm. The full source code and the raw test results are available online~\citeonline{suomela-mf}. An overview of the results is given in Figure~\ref{fig:summary} in the appendix, and selected examples of generator-specific results are shown in Figures \ref{fig:summary-r-large}--\ref{fig:summary-r-desc}. Note that the $y$ axis is linear in these plots.

As we can see from the plots, SortMedian typically performs better than HeapMedian. The running times are consistently low. HeapMedian is a clear winner only for very small window sizes, while SortMedian typically wins by a large factor for larger windows.

One the most interesting findings is shown in Figures \ref{fig:summary-r-asc} and \ref{fig:summary-r-desc}. These plots demonstrate that SortMedian makes a very effective use of partially sorted inputs, while such inputs are actually more difficult for HeapMedian than uniform random inputs. Perhaps the most important factor here is the locality of memory references and cache efficiency.

\section*{Acknowledgements}

Computer resources were provided by the Aalto University School of Science ``Science-IT'' project~\citesoftware{triton}. Many thanks to David Eppstein, Geoffrey Irving, Petteri Kaski, Pat Morin, and Saeed for comments and discussions. This problem has been discussed online on Theoretical Computer Science Stack Exchange \citeonline{suomela-mf1} and Google+ \citeonline{suomela-mf2}.

\setlength{\bibsep}{5pt}
\bibliographystyle{plain}
\bibliography{articles}
{
\renewcommand{\section}{\subsection}
\bibliographystylesoftware{plain}
\bibliographysoftware{software}
\bibliographystyleonline{plain}
\bibliographyonline{online}
}

\newpage
\appendix
\section{Appendix}

\subsection{Python Implementation}\label{app:python}

\begin{verbatim}
def create_array(n):
    return [None] * n

def sort_block(alpha):
    pairs = [(alpha[i], i) for i in range(len(alpha))]
    return [i for v,i in sorted(pairs)]

class Block:
    def __init__(self, h, alpha):
        self.k = len(alpha)
        self.alpha = alpha
        self.pi = sort_block(alpha)
        self.prev = create_array(self.k + 1)
        self.next = create_array(self.k + 1)
        self.tail = self.k
        self.init_links()
        self.m = self.pi[h]
        self.s = h

    def init_links(self):
        p = self.tail
        for i in range(self.k):
            q = self.pi[i]
            self.next[p] = q
            self.prev[q] = p
            p = q
        self.next[p] = self.tail
        self.prev[self.tail] = p

    def unwind(self):
        for i in range(self.k-1, -1, -1):
            self.next[self.prev[i]] = self.next[i]
            self.prev[self.next[i]] = self.prev[i]
        self.m = self.tail
        self.s = 0

    def delete(self, i):
        self.next[self.prev[i]] = self.next[i]
        self.prev[self.next[i]] = self.prev[i]
        if self.is_small(i):
            self.s -= 1
        else:
            if self.m == i:
                self.m = self.next[self.m]
            if self.s > 0:
                self.m = self.prev[self.m]
                self.s -= 1

    def undelete(self, i):
        self.next[self.prev[i]] = i
        self.prev[self.next[i]] = i
        if self.is_small(i):
            self.m = self.prev[self.m]

    def advance(self):
        self.m = self.next[self.m]
        self.s += 1

    def at_end(self):
        return self.m == self.tail

    def peek(self):
        return float('Inf') if self.at_end() else self.alpha[self.m]

    def get_pair(self, i):
        return (self.alpha[i], i)

    def is_small(self, i):
        return self.at_end() or self.get_pair(i) < self.get_pair(self.m)

def sort_median(h, b, x):
    k = 2 * h + 1
    B = Block(h, x[0:k])
    y = []
    y.append(B.peek())
    for j in range(1, b):
        A = B
        B = Block(h, x[j*k:(j+1)*k])
        B.unwind()
        for i in range(k):
            A.delete(i)
            B.undelete(i)
            if A.s + B.s < h:
                if A.peek() <= B.peek():
                    A.advance()
                else:
                    B.advance()
            y.append(min(A.peek(), B.peek()))
    return y
\end{verbatim}

\newpage

\begin{figure}
    \centering
    \includegraphics{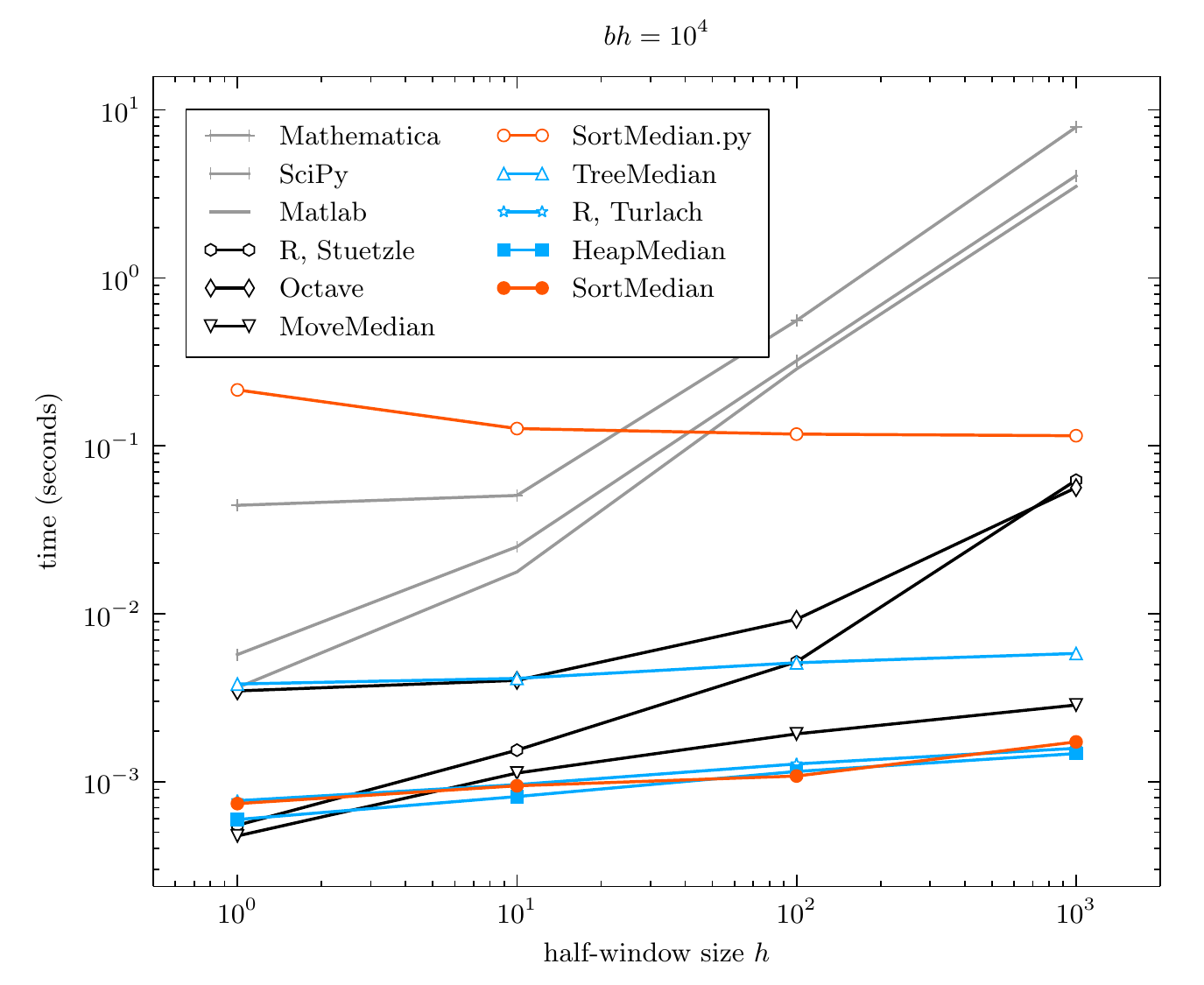}\\
    \includegraphics{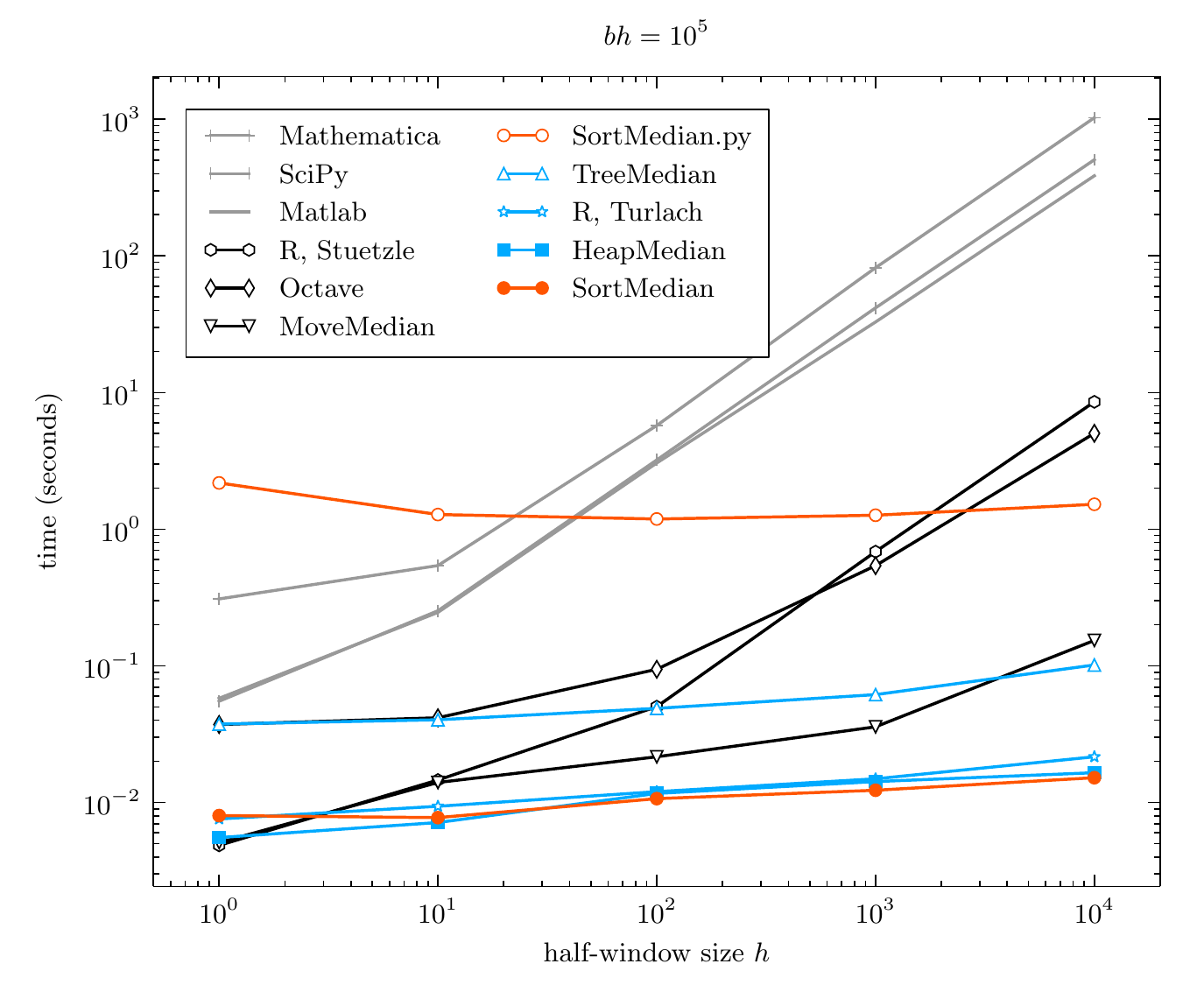}
    \caption{Comparison of all implementations, for $bh = 10^4$ and $bh = 10^5$.}\label{fig:othera}
\end{figure}

\begin{figure}
    \centering
    \includegraphics{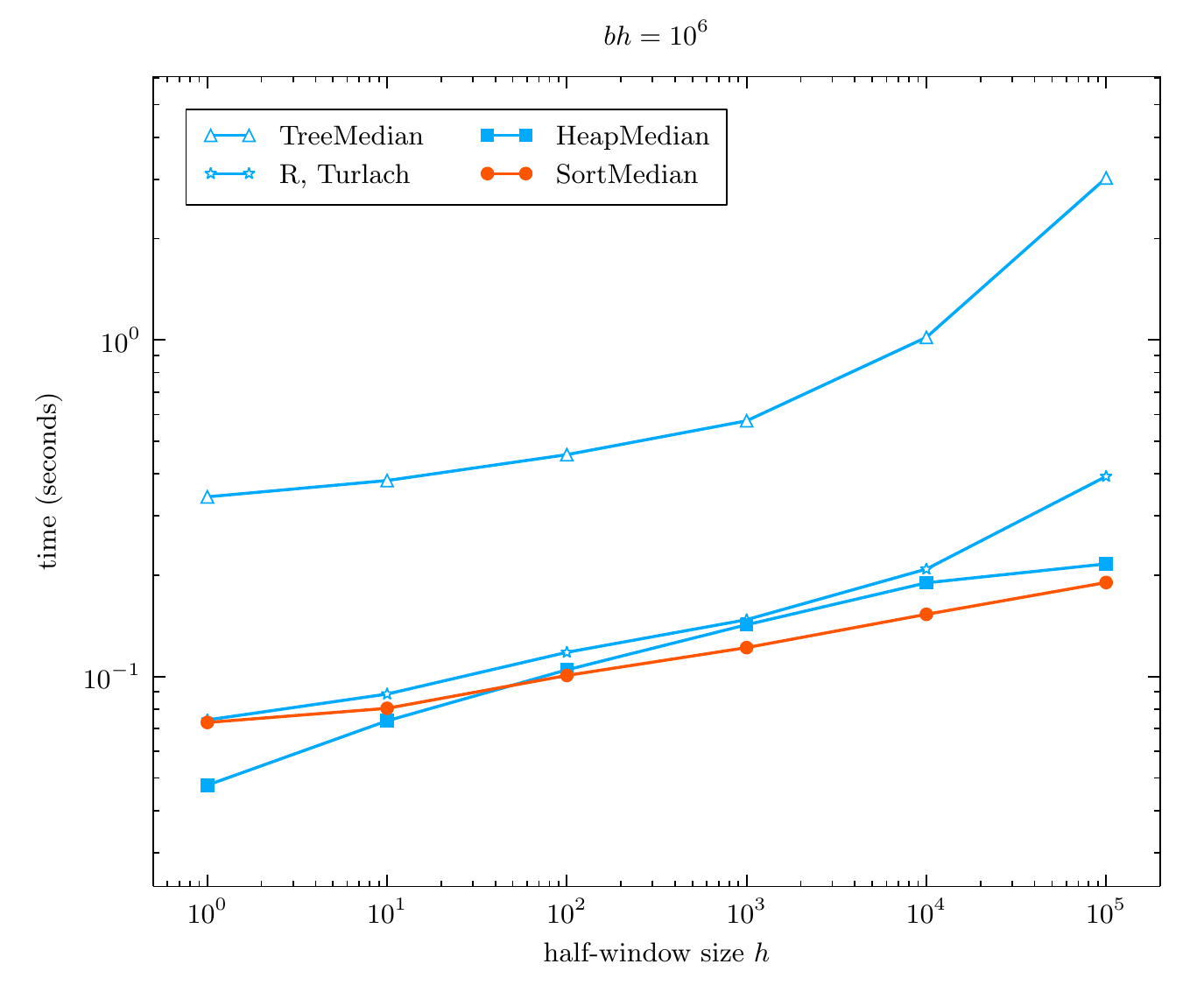}\\
    \includegraphics{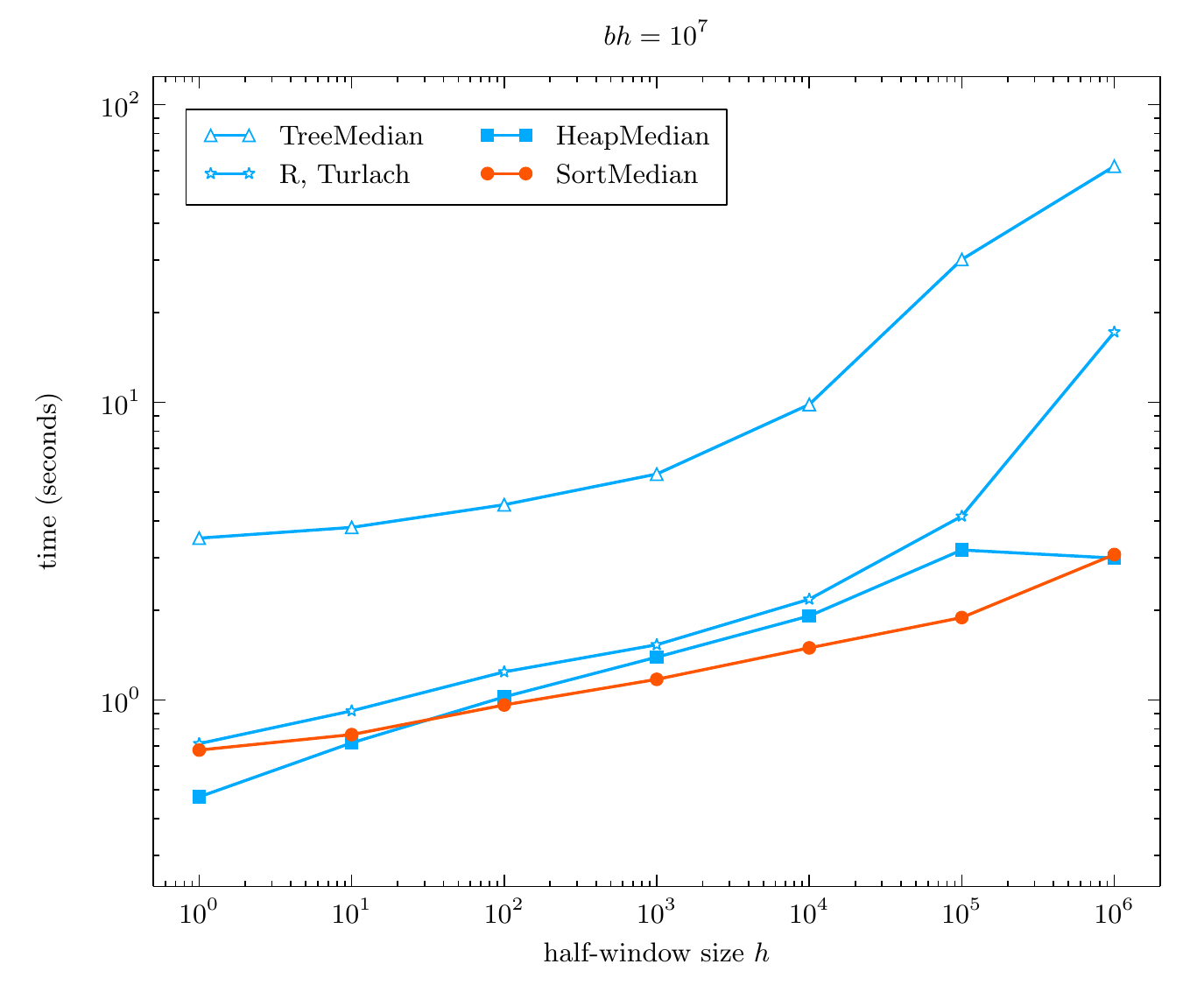}
    \caption{Comparison of the fastest implementations, for $bh = 10^6$ and $bh = 10^7$.}\label{fig:otherb}
\end{figure}

\newcommand{\figexpl}{Solid lines are medians. The shaded area represents 80~\% of the experiments.}

\begin{figure}
    \centering
    \includegraphics{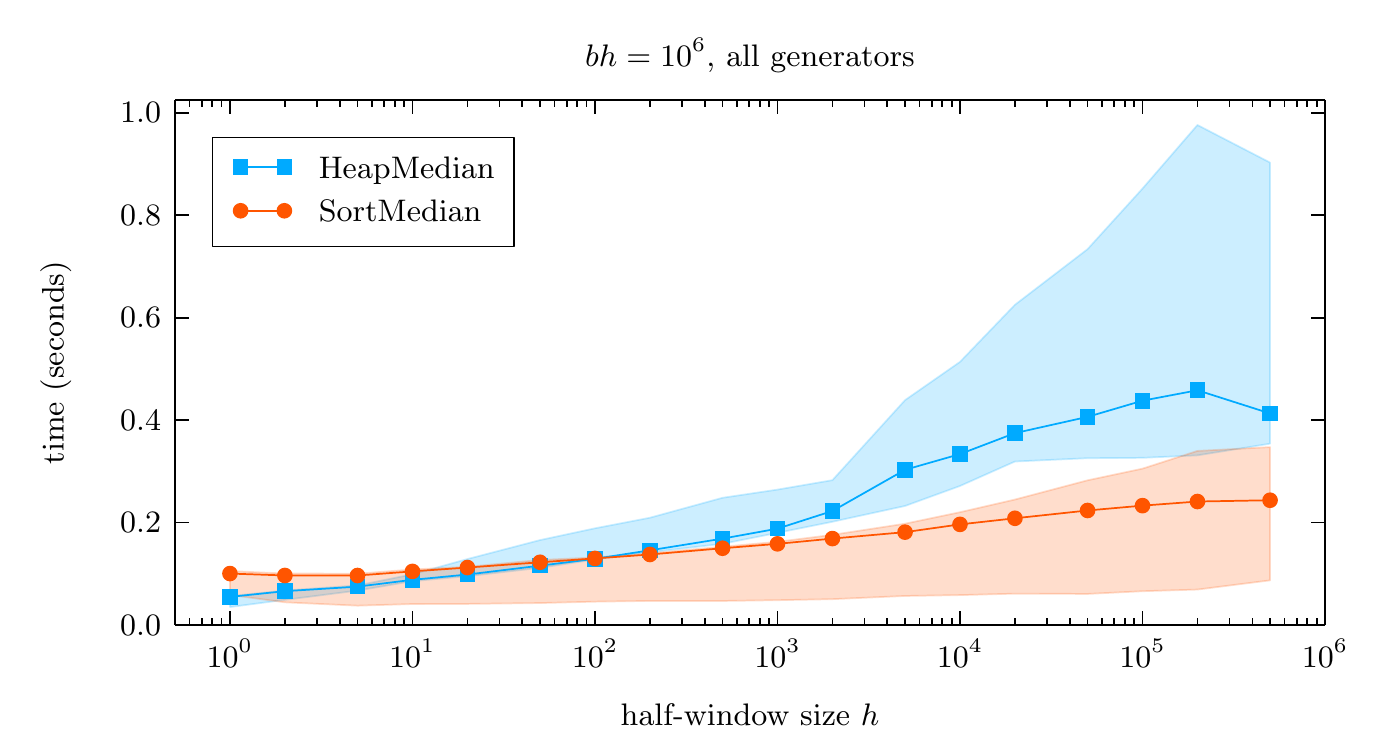}\\
    \includegraphics{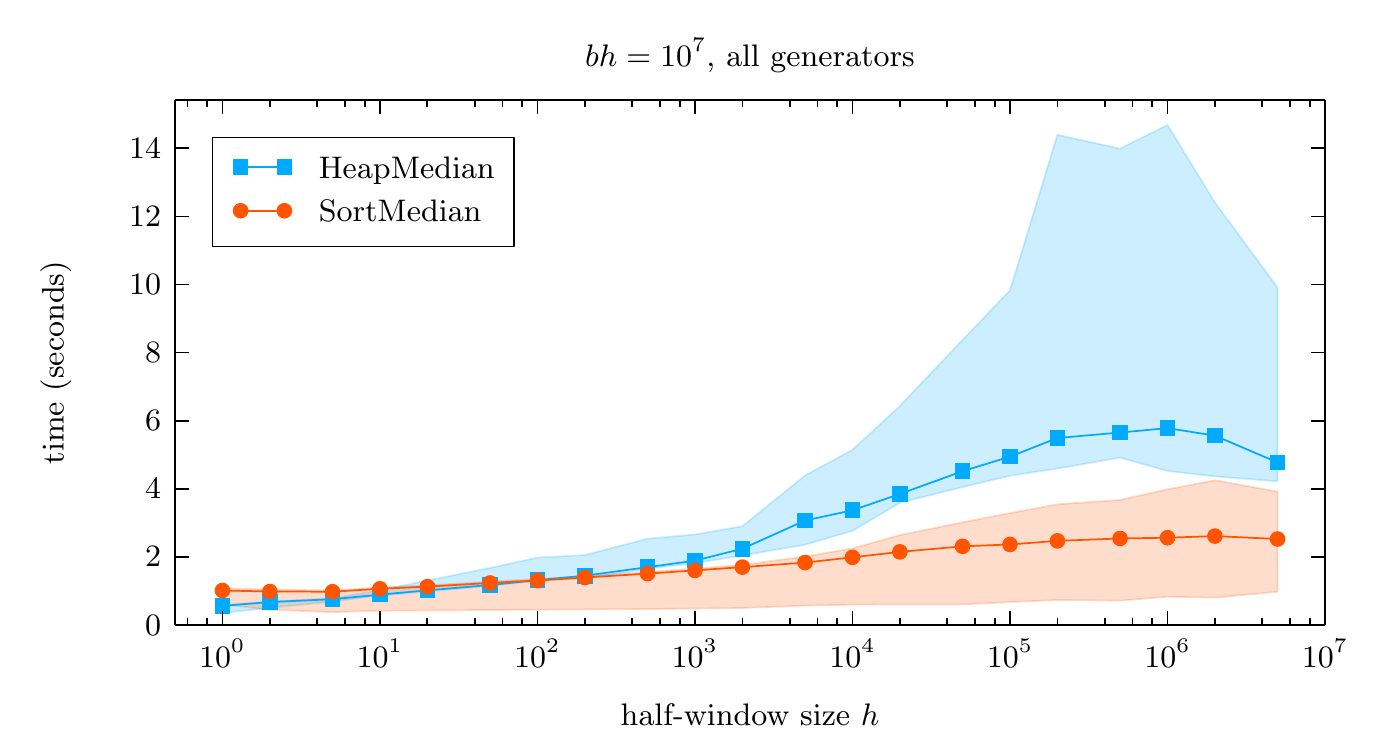}\\
    \includegraphics{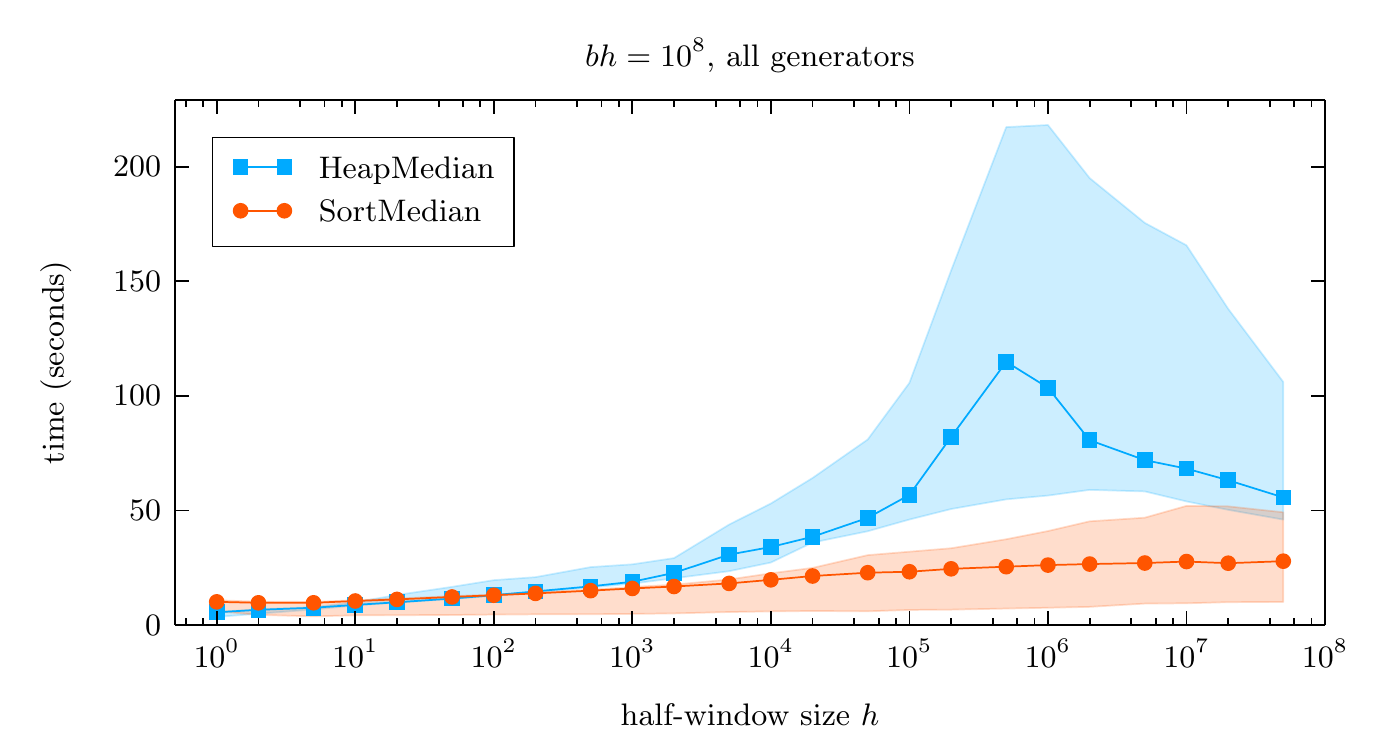}
    \caption{Comparison of SortMedian and HeapMedian. \figexpl}\label{fig:summary}
\end{figure}

\begin{figure}
    \centering
    \includegraphics{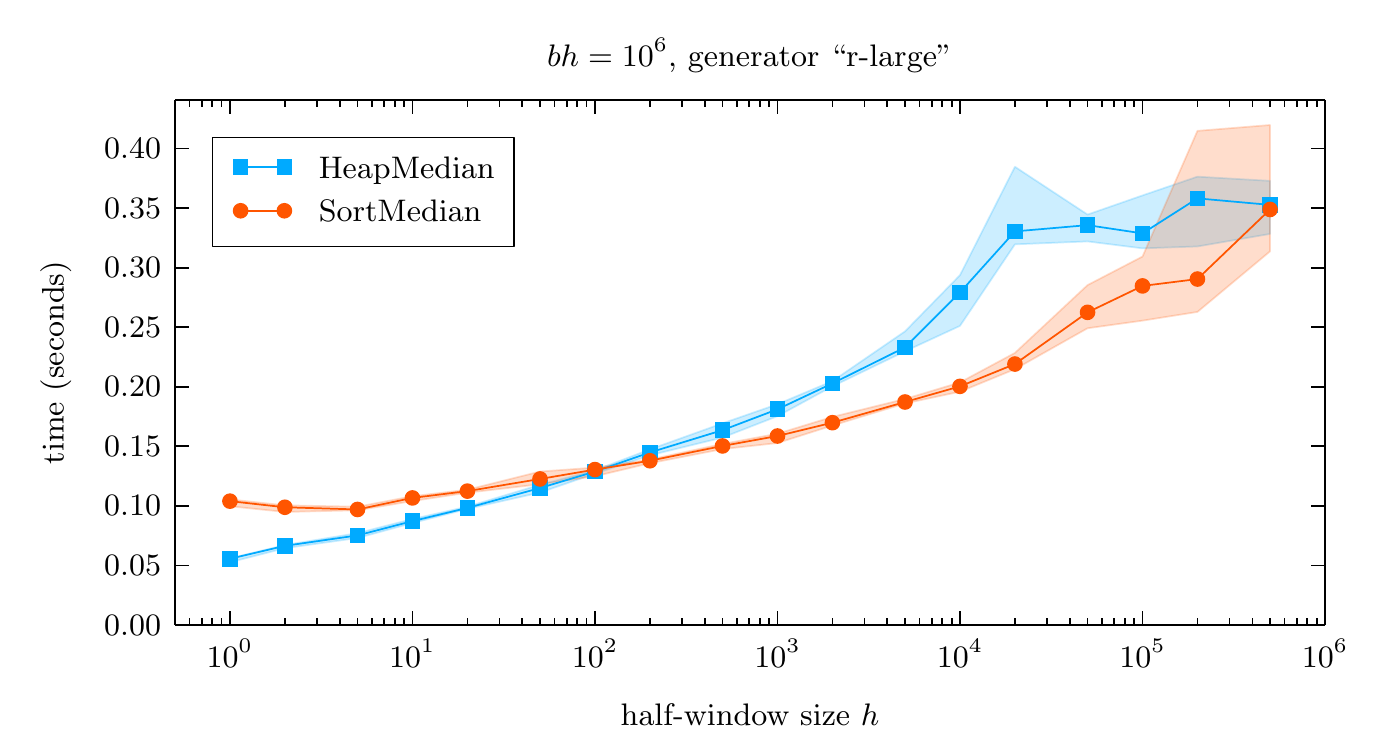}\\
    \includegraphics{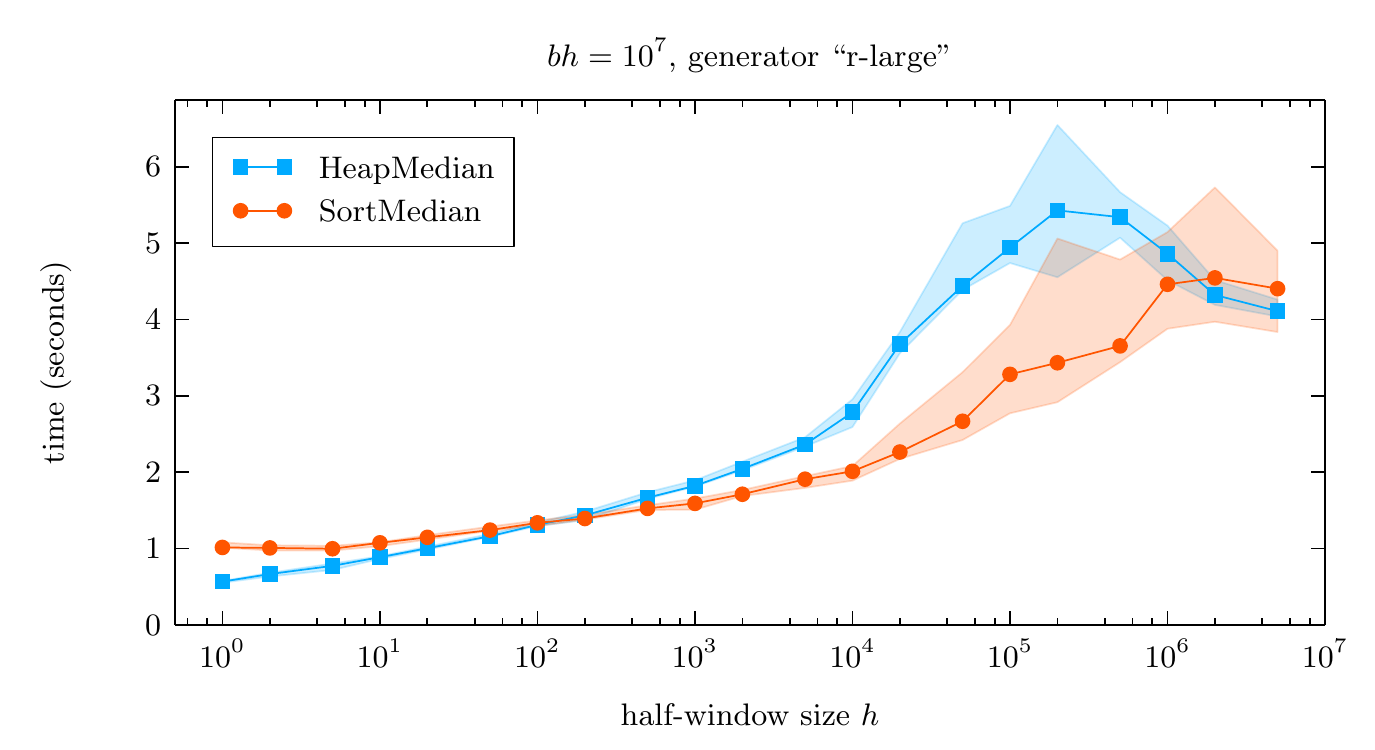}\\
    \includegraphics{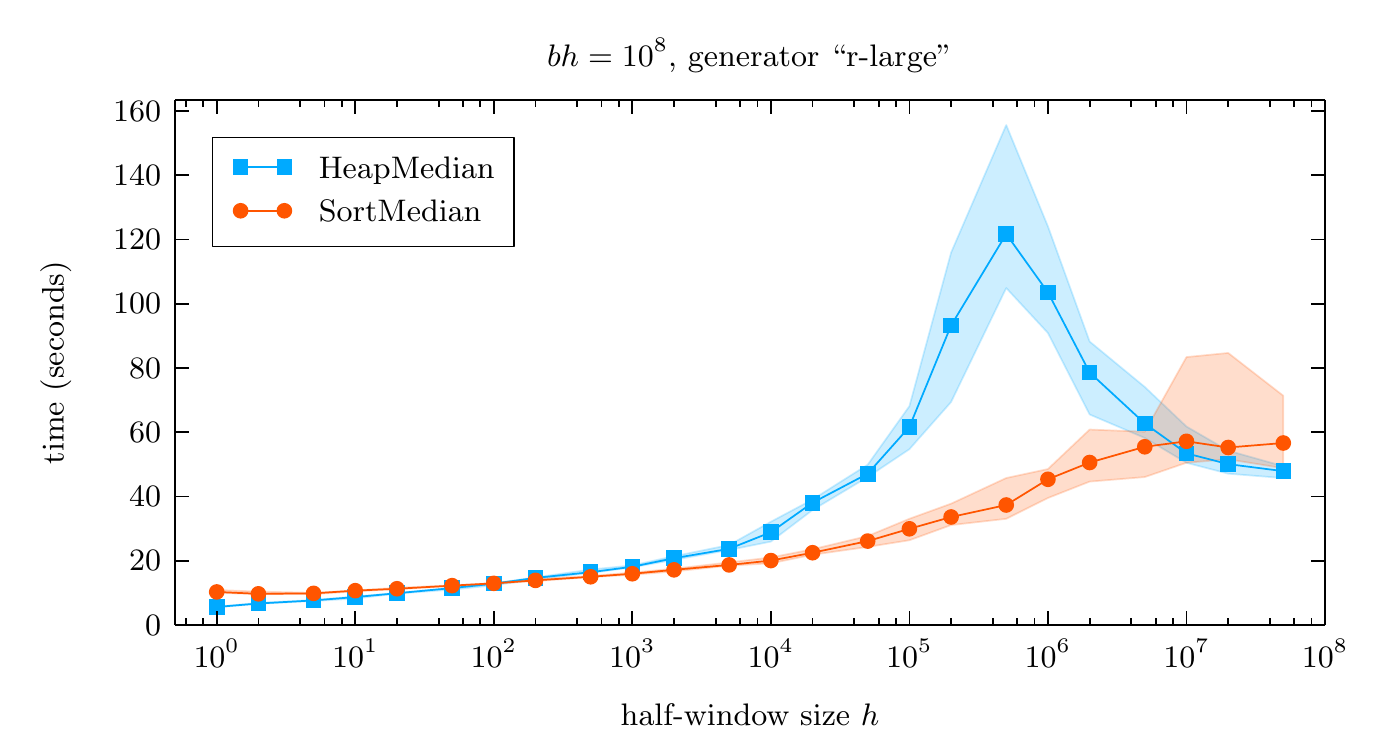}
    \caption{Comparison of SortMedian and HeapMedian for generator \emph{r-large}. \figexpl}\label{fig:summary-r-large}
\end{figure}

\begin{figure}
    \centering
    \includegraphics{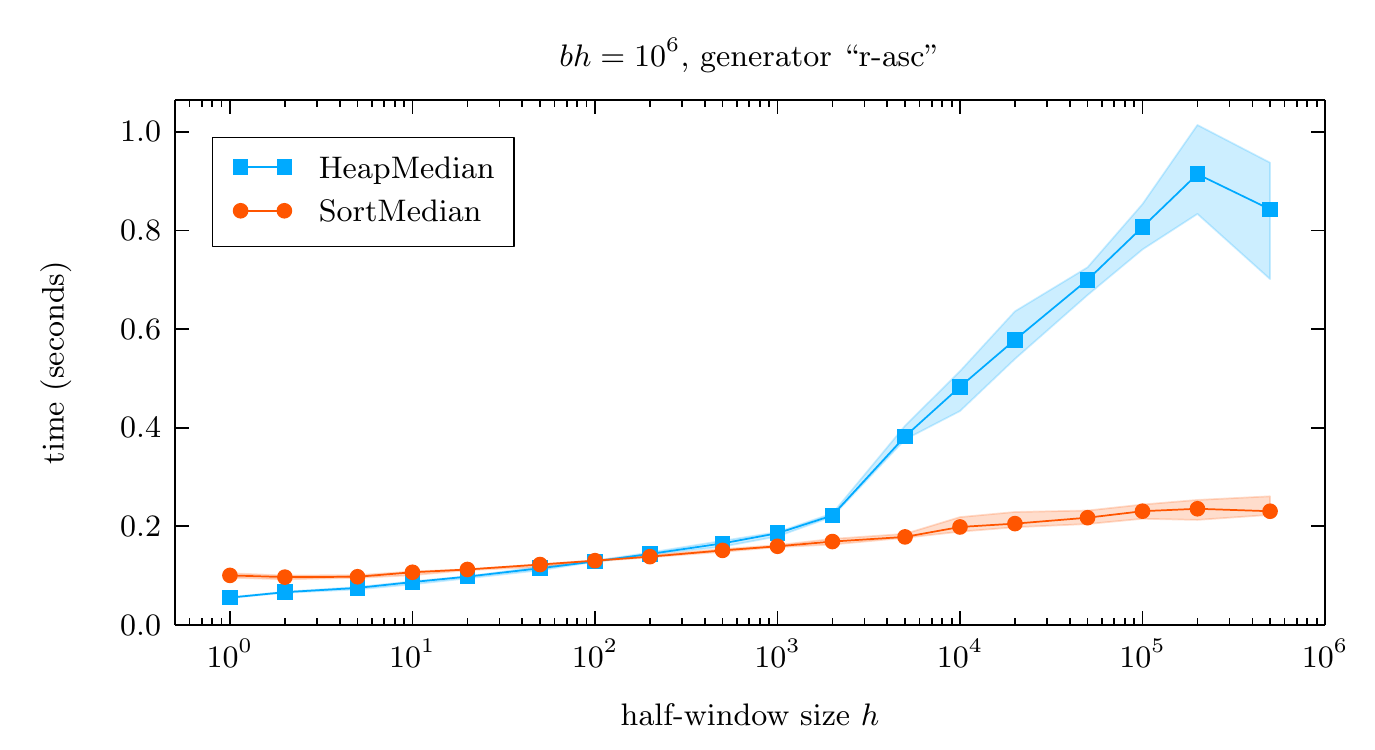}\\
    \includegraphics{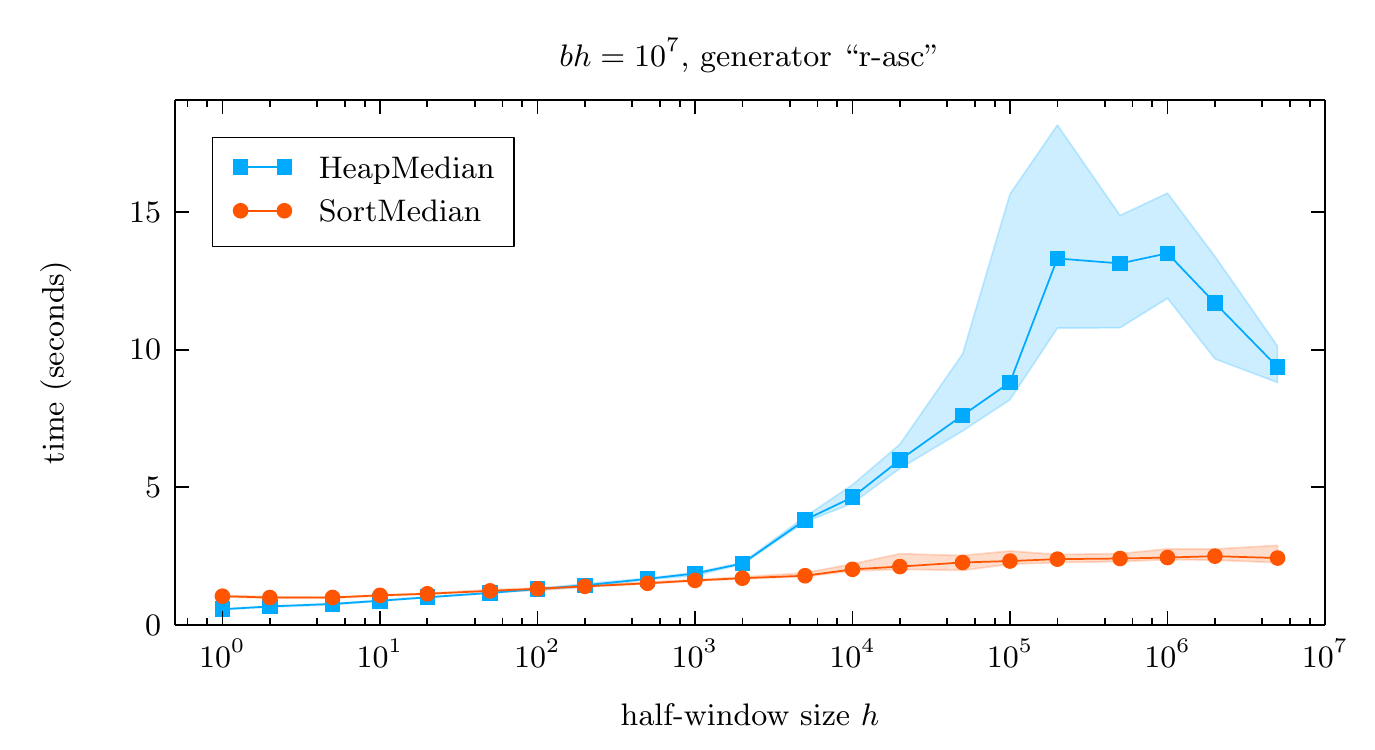}\\
    \includegraphics{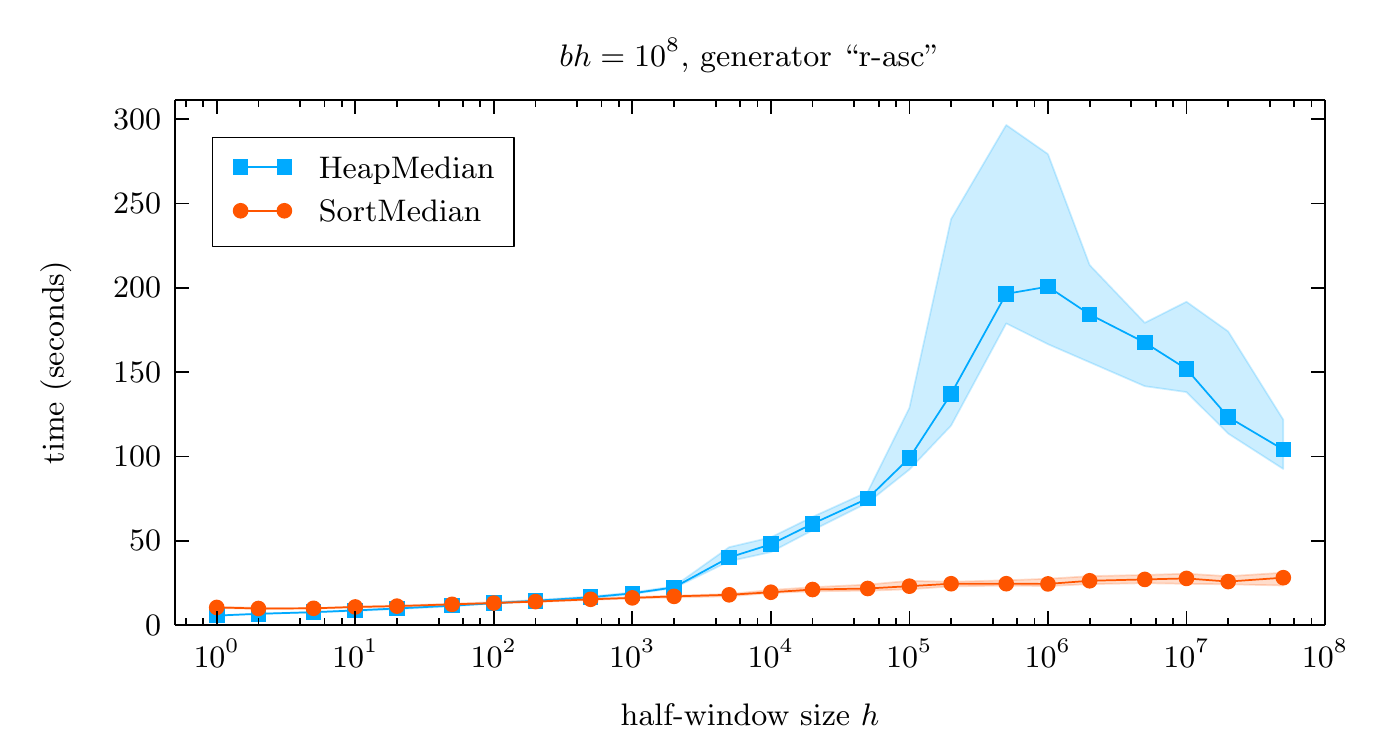}
    \caption{Comparison of SortMedian and HeapMedian for generator \emph{r-asc}. \figexpl}\label{fig:summary-r-asc}
\end{figure}

\begin{figure}
    \centering
    \includegraphics{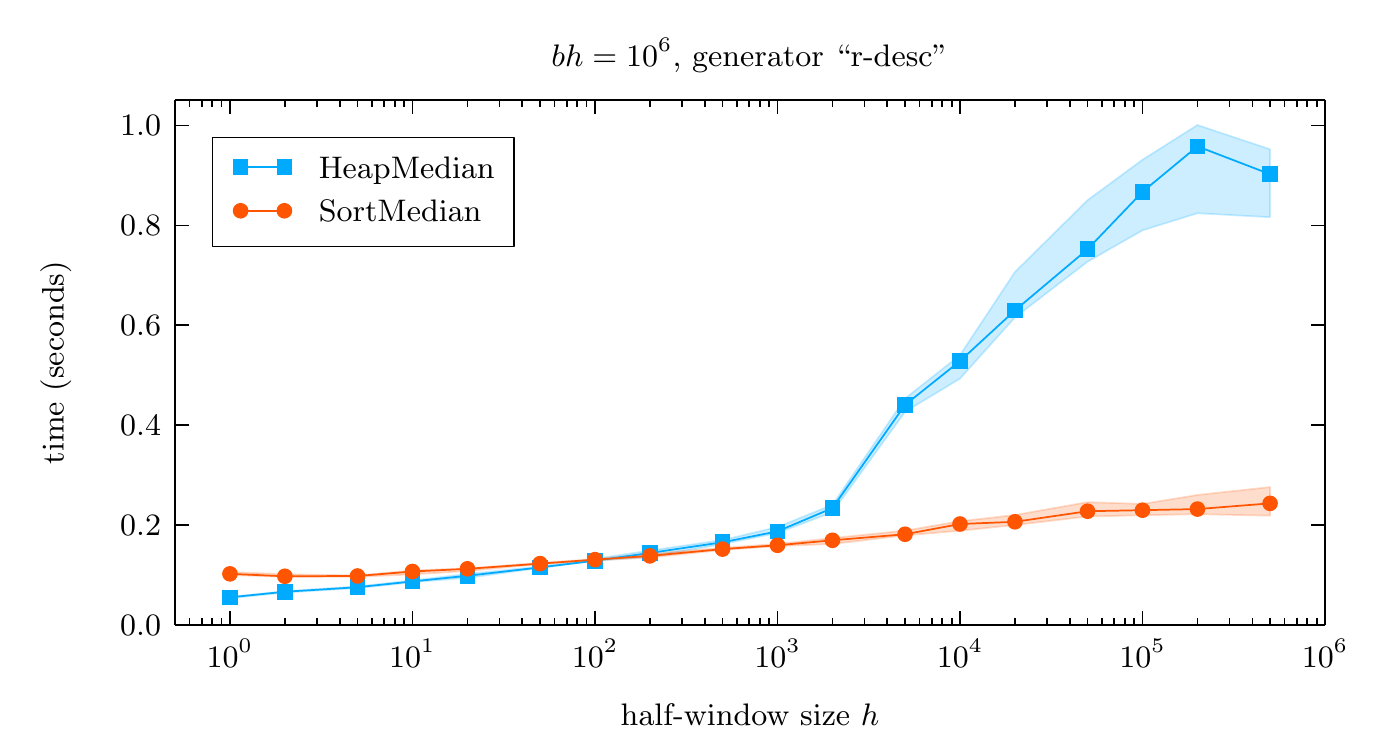}\\
    \includegraphics{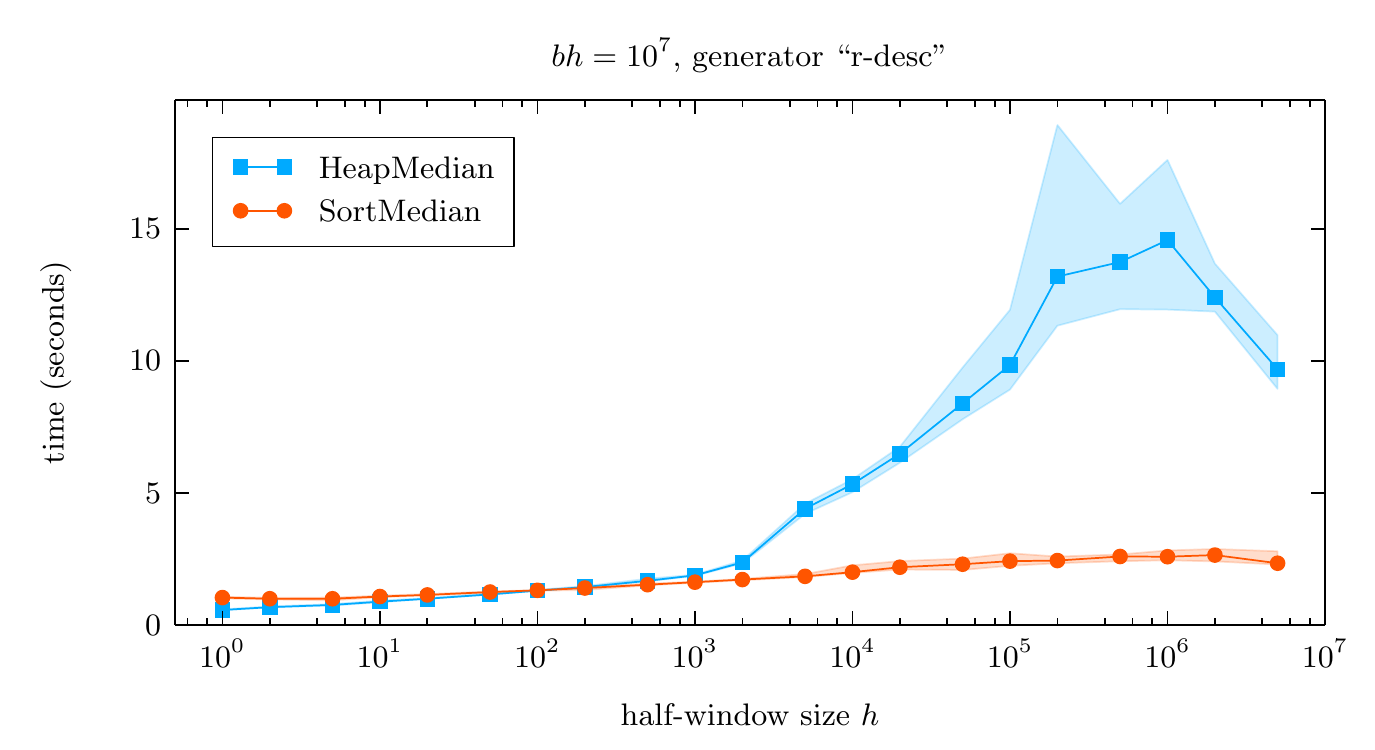}\\
    \includegraphics{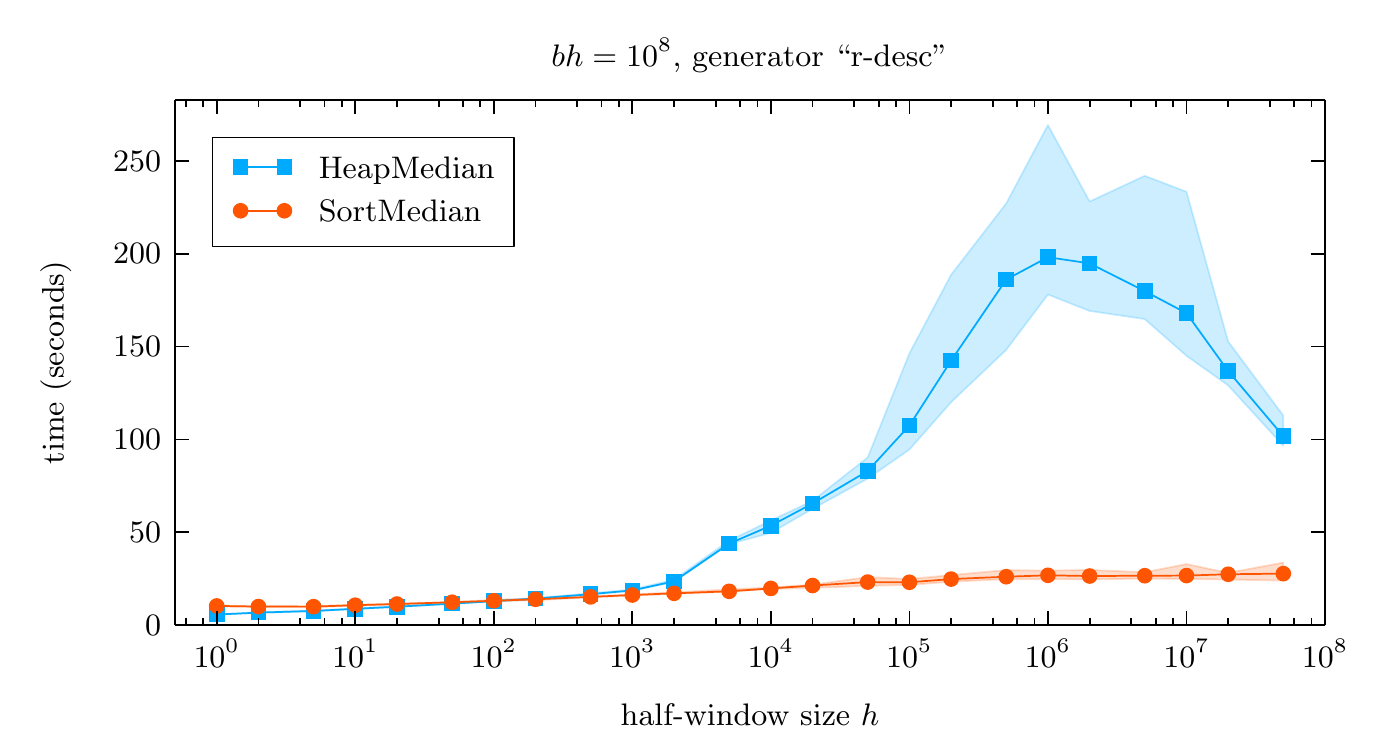}
    \caption{Comparison of SortMedian and HeapMedian for generator \emph{r-desc}. \figexpl}\label{fig:summary-r-desc}
\end{figure}

\end{document}

%% file: block-if.tex
\algorithm{leftmargin=33mm,itemsep=6pt}{
\item[construct($\alpha, \pi$):]
    Return $B = (\alpha,\pi,L,s)$, where: \block{
    $L = \bigl(\alpha[\pi[0]], \alpha[\pi[1]], \dotsc, \alpha[\pi[k-1]]\bigr)$
    
    $s = h$
    }

\item[delete($B,i$):]
    Remove $\alpha[i]$ from $L$

    $s \gets \max \, \{0, s-1\}$

\item[undelete($B,i$):]
    Put $\alpha[i]$ back to $L$

\item[unwind($B$):]
    $\Delete(B,k-1), \Delete(B,k-2), \dotsc, \Delete(B,0)$

\item[advance($B$):]
    $s \gets s + 1$

\item[small($B$):]
    Return $s$

\item[peek($B$):]
    Return the first large element, or $+\infty$ if all elements are small
}

%% file: block-impl.tex
\algorithm{leftmargin=33mm,itemsep=9pt}{
\item[construct($\alpha, \pi$):]
    $p \gets k$

    For each $i = 0, 1, \dotsc, k-1$: \block{
        $q \gets \pi[i]$
        
        $\next[p] \gets q$
        
        $\prev[q] \gets p$

        $p \gets q$
    }
    
    $\next[p] \gets k$

    $\prev[k] \gets p$

    $s \gets h$

    $m \gets \pi[h]$

\item[delete($B,i$):]
    $\next[\prev[i]] \gets \next[i]$

    $\prev[\next[i]] \gets \prev[i]$

    If $\alpha[i]$ was small: \block{
        $s \gets s - 1$
    }
    
    If $\alpha[i]$ was large and $m = i$: \block{
        $m \gets \next[m]$
    }
    
    If $\alpha[i]$ was large and $s > 0$: \block{
        $m \gets \prev[m]$

        $s \gets s - 1$
    }

\item[undelete($B,i$):]
    $\next[\prev[i]] \gets i$

    $\prev[\next[i]] \gets i$

    If $\alpha[i]$ is small: \block{
        set $m \gets \prev[m]$
    }

\item[unwind($B$):]
    For each $i = k-1, k-2, \dotsc, 0$: \block{
        $\next[\prev[i]] \gets i$

        $\prev[\next[i]] \gets i$
    }

    $m \gets k$, $s \gets 0$

\item[advance($B$):]
    $m \gets \next[m]$

    $s \gets s + 1$

\item[small($B$):]
    Return $s$

\item[peek($B$):]
    Return $\alpha[m]$
}

%% file: postprocess.tex
\algorithm{leftmargin=39mm}{
\item[postprocess($X,P$):]
    $B \gets \Construct(X_0, P_0)$
    
    Print $\Peek(B)$\codelabel{\stepa}
    
    For each $j = 1, 2, \dotsc, b-1$: \block{
        $A \gets B$
        
        $B \gets \Construct(X_j, P_j)$
        
        $\Unwind(B)$
        
        For each $i = 0, 1, \dotsc, k-1$:\block{
            $\Delete(A, i)$\codelabel{\stepb}
            
            $\Undelete(B,i)$
            
            If $\Small(A) + \Small(B) < h$:
            \block{
                If $\Peek(A) \le \Peek(B)$: \block{
                    $\Advance(A)$
                }
                
                Otherwise: \block{
                    $\Advance(B)$
                }
            }
            
            Print $\min \{ \Peek(A), \Peek(B) \}$\codelabel{\stepc}
        }
    }
}

%% file: invariant.tex
\newcommand{\myhead}[1]{\multicolumn{6}{@{}l@{}}{\textbf{\boldmath #1}} \\ \addlinespace}
\begin{tabular}{@{}llllll@{}}
\toprule
    Partial order &
    $H_{\mydot A \mydot B}$ &
    $S_{\mydot A}$ &
    $S_{\mydot B}$ &
    $S_{\myddot A}$ &
    $S_{\myddot B}$
    \\
\midrule
    \myhead{$a_A < p_A$}
    $a_B < p_B$ &
    $H_{AB} - a_A + a_B$ &
    $S_A - a_A$ &
    $S_B - p_B + a_B$ &
    $S_{\mydot A}$ &
    $S_{\mydot B} + p_B$
    \\
    $p_B < a_B < \min \{q_A, q_B\}$ &
    $H_{AB} - a_A + a_B$ &
    $S_A - a_A$ &
    $S_B$ &
    $S_{\mydot A}$ &
    $S_{\mydot B} + a_B$
    \\
    $q_B < \min \{q_A, a_B\}$ & 
    $H_{AB} - a_A + q_B$ &
    $S_A - a_A$ &
    $S_B$ &
    $S_{\mydot A}$ &
    $S_{\mydot B} + q_B$
    \\
    $q_A < \min \{q_B, a_B\}$ &
    $H_{AB} - a_A + q_A$ &
    $S_A - a_A$ &
    $S_B$ &
    $S_{\mydot A} + q_A$ &
    $S_{\mydot B}$
    \\
    \midrule
    \myhead{$p_a \le a_A$}
    $p_B < a_B$ &
    $H_{AB}$ &
    $S_A - p_A$ &
    $S_B$ &
    $S_{\mydot A} + p_A$ &
    $S_{\mydot B}$
    \\
    $a_B < p_B < p_A$ &
    $H_{AB} - p_A + a_B$ &
    $S_A - p_A$ &
    $S_B - p_B + a_B$ &
    $S_{\mydot A}$ &
    $S_{\mydot B} + p_B$
    \\
    $\max \{ a_B, p_A \} < p_B$ &
    $H_{AB} - p_B + a_B$ &
    $S_A - p_A$ &
    $S_B - p_B + a_B$ &
    $S_{\mydot A} + p_A$ &
    $S_{\mydot B}$
    \\
\bottomrule
\end{tabular}

%% file: versions.tex
\begin{tabular}{@{}l@{\qquad}l@{\qquad}l@{\qquad}l@{}}
\toprule
Software
& Function
& Versions
& Platform
\\
\midrule
\emph{R} \citesoftware{r}
& runmed \citesoftware{r-manual}
& R 3.1.0
& OS X
\\
\addlinespace
\emph{Octave} \citesoftware{octave}
& medfilt1 \citesoftware{octave-mf}
& GNU Octave 3.8.1
& OS X
\\
&
& signal 1.3.0
&
\\
\addlinespace
\emph{Matlab} \citesoftware{matlab}
& medfilt1 \citesoftware{matlab-mf}
& Matlab R2014a (8.3.0.532)
& Linux
\\
\addlinespace
\emph{SciPy} \citesoftware{scipy}
& scipy.signal.medfilt \citesoftware{scipy-mf}
& Python 2.7.7
& OS X
\\
&
& numpy 1.8.1
&
\\
&
& scipy 0.14.0
&
\\
\addlinespace
\emph{Mathematica} \citesoftware{mathematica}
& MedianFilter \citesoftware{mathematica-mf}
& Mathematica 9.0.1.0
& OS X
\\
\midrule
\multicolumn{4}{@{}l@{}}{\small OS X: Intel Core i7, 1.7 GHz, 8 GB RAM.} \\
\multicolumn{4}{@{}l@{}}{\small Linux: Intel Xeon, 2.8 GHz, 256 GB RAM.} \\
\bottomrule
\end{tabular}